\UseRawInputEncoding
\pdfoutput=1
\documentclass[10pt,pra,aps,twocolumn]{revtex4-1}
\usepackage{amsmath,bbm}
\usepackage{amsthm}
\usepackage{latexsym}
\usepackage{amssymb}
\usepackage{float}
\usepackage{graphicx}           
\usepackage{color}
\usepackage{mathpazo}
\usepackage{relsize}
\usepackage{braket}
\usepackage[colorlinks=true,linkcolor=blue,citecolor=magenta,urlcolor=blue]{hyperref}
\usepackage{changes}
\usepackage{cancel}
\usepackage{xcolor}

\usepackage[utf8]{inputenc}
\usepackage[T1]{fontenc}

\newcommand{\be}{\begin{equation}} 
\newcommand{\ee}{\end{equation}}
\newcommand{\beq}{\begin{eqnarray}}
\newcommand{\eeq}{\end{eqnarray}}

\def\squareforqed{\hbox{\rlap{$\sqcap$}$\sqcup$}}
\def\qed{\ifmmode\squareforqed\else{\unskip\nobreak\hfil
\penalty50\hskip1em\null\nobreak\hfil\squareforqed
\parfillskip=0pt\finalhyphendemerits=0\endgraf}\fi}
\def\endenv{\ifmmode\;\else{\unskip\nobreak\hfil
\penalty50\hskip1em\null\nobreak\hfil\;
\parfillskip=0pt\finalhyphendemerits=0\endgraf}\fi}

\newcommand{\I}{\mathbbm{1}}

\newcommand{\ra}{\rangle}
\newcommand{\la}{\langle}

\newcommand{\cH}{{\cal H}}

\makeatletter
\newtheorem*{rep@theorem}{\rep@title}
\newcommand{\newreptheorem}[2]{%
\newenvironment{rep#1}[1]{%
 \def\rep@title{#2 \ref{##1}}%
 \begin{rep@theorem}}%
 {\end{rep@theorem}}}
\makeatother

\def\tr{\mbox{tr}}
\newtheorem{thm}{Theorem}%[section]
\newreptheorem{thm}{Theorem}
\newtheorem{lemma}{Lemma}

\newtheorem{coro}{Corollary}
\newtheorem{defi}{Definition}

\begin{document}

\title{Scalable noncontextuality inequalities and certification of multiqubit quantum systems}

%\title{\textbf{Remik: The title is too strong}  Scalable Self-testing of Qubit Graph %States based on Noncontextuality Inequalities}

\author{Rafael Santos}
\author{Chellasamy Jebarathinam}
\author{Remigiusz Augusiak}
\affiliation{Center for Theoretical Physics, Polish Academy of Sciences, Aleja Lotnik\'{o}w 32/46, 02-668 Warsaw, Poland}

\begin{abstract}
We propose a family of noncontextuality inequalities
and show that they can be used for certification of multiqubit quantum systems.
Our scheme, unlike those based on non-locality, does not require spatial separation between the subsystems, yet it makes use of certain compatibility relations between measurements. Moreover, it is scalable in the sense that the number of expectation values that are to be measured to observe contextuality scales polynomially with the number of qubits that are being certified. In a particular case we also show our scheme to be robust to errors and experimental imperfections. Our results seem promising as far as certification of physical set-ups is concerned in scenarios  where spatial separation between the subsystems cannot be guaranteed. 

\end{abstract}

\maketitle

\section{Introduction}

Multiqubit entangled states constitute a key resource for various quantum information tasks such as quantum  computation \cite{RB01,KMN+07} and error correction \cite{Sho95,Ter15},  quantum communication \cite{BPS+20,HBE21}, quantum simulations \cite{LWG+10,MDN+11}, and cryptographic protocols \cite{GKB+07,EKM+17}. To realize genuine quantum technologies employing such tasks, the back-end user should be guaranteed that the quantum devices work as specified by the provider. The standard state verification schemes based on quantum tomography \cite{TNW+02,KSW+05}, however, suffer from two problems: they are unfeasible for larger systems and require using trusted and well characterized measuring devices.

Observing nonclassical correlations through the violation of a Bell-type inequality \cite{Bel64} can be used to detect entanglement in a device-independent way, i.e., it implies the presence of entanglement without the need to have a trust in the measurement devices. This property of the violation of Bell inequalities makes them a useful resource for implementing quantum information protocols such as quantum key distribution in a device-independent way \cite{BCP+14}. Remarkably, maximal quantum violation of certain Bell inequalities can be used to demonstrate a phenomenon called  "self-testing of quantum states and measurements"  \cite{MY04,SB20}, which can be used to provide device-independent characterization of quantum devices.
Recently, such form of certification based on the phenomenon of nonlocality has been explored extensively. For instance, several self-testing statements
for multiqubit graph states have recently been derived in Refs. \cite{Kan16,BAS+20,PPW21}.
However, genuine demonstration of the violation of Bell inequalities requires a spatial separation between the subsystems.

Sequential quantum measurements on a single system can be used to observe 
quantum contextuality \cite{KS67} and temporal quantum correlations \cite{LG85,BE14,BKM+15} through the violation of suitable inequalities. 
Apart from the foundational relevance of these two notions of nonclassicality, on one side, they have been explored as a resource for quantum information applications such as measurement-based quantum computation \cite{BCG+21,MPA+14,PhysRevLett.119.120505,2014contextualitymagicstates}.  
On the other side, they have also been used for certification of relevant quantum properties such as the dimension of the underlying quantum system \cite{GCC+14,SSG+20}. More importantly, contextuality and temporal correlations have also been exploited for
certification of quantum states and/or measurements \cite{BRV+21,BRV+19,IMO+20,SSA20,MMJ+21}. 

%\textbf{It can also and can give answers to questions from a foundational perspective %of quantum theory \cite{PhysRevLett.125.060406}. }

Motivated by the above results, in this work we 
introduce a family of noncontextuality inequalities
that are maximally violated by many-qubit quantum systems and 
certain pairs of anticommuting observables. In constructing our inequalities 
we exploit the multiqubit stabilizer formalism known for its use in quantum error
correction \cite{GottesmanThesis,PhysRevA.65.012308,PhysRevA.78.042303}.
These inequalities are scalable in the sense that the number of expectation values they are built from scales polynomially with the number of observables $2n$ that are measured; yet their maximal violation can be achieved by quantum systems of dimension at least $2^n$. From this point of view they can be seen as dimension witnesses.
In the particular case $n=3$ our family reproduces an inequality that in 
the non-locality context is known 
as the Mermin-Ardehali-Belinskii-Klyshko (MABK) inequality \cite{PhysRevA.46.5375,Uspekhi,Mer90} 
(see also Refs. \cite{GTH+05,TGB06, WMA13} for other approaches to reveal 
Bell nonlocality/quantum contextuality based on stabilizer formalism).
Yet for $n>3$ these families are distinct. We then show that 
our inequalities can be used for certification of multiqubit quantum systems
in the sense of Ref. \cite{IMO+20}. In fact, we generalize the results
of that work to any number of qubits.

Our work is organized as follows. In Section \ref{preliminaries} we outline the  contextuality scenario and provide the definitions of graph states and self-testing. In Section \ref{simplestinequality} we present the simplest inequality designed to certify the three qubit-graph state corresponding to the complete graph together with three pairs of anticommuting observables. In Section \ref{scalable} we present a scalable family of inequalities designed to certify multiqubit quantum systems. In Section \ref{Robustness} we investigate whether our certification schemes are robust.

\section{Preliminaries}\label{preliminaries}

We begin by illustrating our scenario and introducing the relevant notations and definitions.

\subsection{Contextuality scenario}
\label{scenario}

A standard contextuality scenario is defined by a triple of sets: a set of measurements, a set of outcomes of the measurements, and a set of contexts, which are the subsets of compatible measurements. The notion of compatibility in a contextuality scenario means that the measurements that belong to the same context can be performed jointly or in a sequence in such a way that the observed statistics are independent of the
order in which these measurements were performed.
%
%meaning that the measurement does not destroy the system in a such way that after many %runs of the experiment the statistics returned will be the same independently of the %order of the performed measurements. 
%
In the latter case, however, the measurements are non-demolishing,
meaning that they do not physically destroy the system.

Each run of the experimental observation comprises of preparation of a physical system
followed by a sequence of non-demolishing measurements in a device as depicted in Fig. \ref{fig}. The measurement device has no memory and returns only the actual post-measurement state. 
%
%\deleted{If this assumption is not satisfied, 
%any measurement statistics produced in the set-up can be reproduced classically. 
%Therefore, the above assumption on the measurement device is needed to demonstrate %quantum contextuality using the 
%sequential measurement set-up.} 
%
The measurement device has different settings, each of which yields two outcomes which we label by $\pm1$. The contexts will be defined in the specific scenarios studied.
Let us stress here that in the quantum case the above scenario comprises the most general situation in which the physical system is described  by a mixed state and the measurements need not be projective.

After repeating this experiment many times, one estimates the joint probabilities of obtaining the outcomes of measurements that are performed on the preparation, and consequently, their correlation functions, which are average values over the outcomes of the measurements. For instance, if the measurements $A_1,A_2,A_3$, which belong to the same context, are performed in sequence or jointly, we can estimate the $2^3$ joint probabilities $p(a_1,a_2,a_3|A_1,A_2,A_3)$ as well as the correlation function
\beq
\la A_1A_2A_3 \ra=\sum_{a_i=\pm1} a_1a_2a_3 p(a_1,a_2,a_3|A_1,A_2,A_3).
\eeq
This notation can be naturally extended to any sets of compatible
measurements.

\begin{figure}[http]
\centering
\includegraphics[scale=0.39]{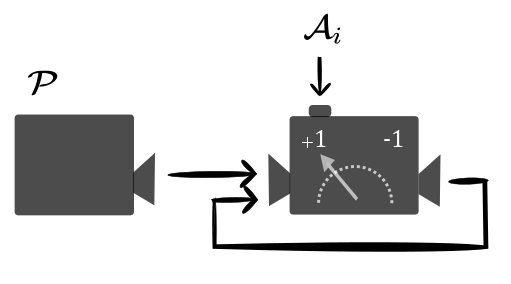}
\caption{{\bf Measurement set-up.} A contextuality experiment comprises of a preparation device $\mathcal{P}$ that prepares some quantum state $\rho$ which is later measured sequentially by the non-demolishing measurement device
with settings $\mathcal{A}_i$; each of these measurements yields the $\pm 1$ outcome. Figure from Ref. \cite{SSA20}.}
\label{fig}
\end{figure}

To reveal contextuality in the experiment one typically uses noncontextuality inequalities as violation of such inequalities by the joint probabilities implies that any noncontextual hidden variable model cannot reproduce them \cite{KS67}. Typically 
such inequalities are defined in terms of linear expressions composed of correlation functions. For instance, in a scenario where the measurements are performed in triples, we can consider the following form of inequality:
\beq \label{gennoncineq}
\mathcal{I} := \sum_{i,j,k} c_{i,j,k} \la A_i A_j A_k \ra \le \eta_C \le \eta_Q
\eeq
where $c_{i,j,k}$ are real coefficients to be chosen, and $\eta_C$ and $\eta_Q$ are the classical and quantum bounds.

If there is a noncontextual hidden variable model that describes the joint probabilities, then the inequality with the classical bound $\eta_C$ is satisfied. Here, the meaning of classicality is mathematically defined as the existence of a noncontextual hidden variable model, for which the expectation values
in (\ref{gennoncineq}) factorize and each individual expectation value is deterministic, i.e., $\la A_i A_i A_k \ra =a_ia_ja_k$, with $a_i\in\{+1,-1\}$. Consequently, the classical bound $\eta_C$ of (\ref{gennoncineq}) can be derived as the maximum value that can be attained by any such model,
\beq
\eta_C = \max_{a_i=\pm1}\left( \sum_{i,j,k} c_{i,j,k}  a_ia_ja_k \right).
%\max_{a_i,a_j,a_k=\{ \pm 1\}} 
\eeq
%The notion of classicality captured by this expression has a quite similar meaning of the notion of classicality in Kochen-Specker contextuality. 

On the other hand, the quantum bound of the inequality defined as the optimal value of the linear expression obtained over all the possible quantum states and measurements in any Hilbert space. Since we do not specify the dimension of the underlying Hilbert space, without any loss of generality, we can assume that the measurements are projective and the state is pure (see, e.g., Ref. \cite{IMO+20} where an extension of the Neumark dilation theorem is proven). In other words, any correlations obtained within the above experiment can always be reproduced with a pure state and projective measurements satisfying the same compatibility relations.

For instance, in the case of the inequality given in \eqref{gennoncineq} the quantum bound is evaluated to be
\beq \label{tracedefined}
\eta_Q = \sup_{A_i;\rho} \left[ \sum_{i,j,k} c_{i,j,k} \tr\left( \rho  A_iA_jA_k \right) \right],
\eeq
where the observables $A_i$ are Hermitian operators acting on a  Hilbert space $\mathcal{H}$ and satisfying $A^2_i=\I$ for any $i$ and $\rho=|\psi\rangle\!\langle\psi|$ is some pure state that describes the preparation. 

%\textcolor{red}{Note that since the state $\rho$ is a rank-1 projector, the trace in the %Eq.\eqref{tracedefined} is well defined even in the case when the Hilbert space %$\mathcal{H}$ is infinite dimensional because the eigenvalues of the measurement %operators $A_i$ is limited to $\pm1$.}

%Since $\rho$ can be chosen to be pure, i.e. $\rho = | \psi \ra \la \psi |$, finding %the quantum bound is equivalent to finding the largest eigenvalue of the quantum %operator corresponding to the given expression of the noncontextuality inequality.
Our aim here is to introduce certain noncontextuality inequalities
that are scalable in the sense that the number of expectation values they consists of grows polynomially with $n$, and, at the same time their maximal violation can be achieved only by quantum systems of dimension $2^n$. We also explore whether these inequalities can be used for certification of quantum states and measurements.

\subsection{Graph states}

A graph $\mathcal{G}=(\mathcal{V},\mathcal{E})$ is a mathematical object defined by a set of vertices $\mathcal{V}$ and a set of edges $\mathcal{E}$ that connect some pairs of vertices. By $\mathcal{N}(i)$ we denote the neighborhood of the vertex $i$, that is, a set of those vertices that are connected to $i$ by an edge. Also, we call a graph connected if any two vertices are connected by a path composed of edges. 

Interestingly, one can exploit the notion of a graph to 
define classes of pure entangled multipartite states. 
While in principle there are many ways of doing that 
here we follow the definition based on the 
stabilizer formalism \cite{Got96} (see also Ref. \cite{HDE+06} for a review on graph states). It allows one to associate an $N$-qubit entangled state to any connected $N$-vertex graph $\mathcal{G}$. 

In order to present the construction consider the Pauli matrices 
\be
X =  \begin{pmatrix} 
  0 & 1 \\ 
  1 & 0 \\
  \end{pmatrix},
\quad 
Y =  \begin{pmatrix} 
  0 & -\mathbbm{i} \\ 
  \mathbbm{i} & 0 \\
  \end{pmatrix},
  \quad
Z =  \begin{pmatrix} 
  1 & 0 \\ 
  0 & -1 \\
  \end{pmatrix}.
\ee
Now, to each vertex $i\in\mathcal{V}$ one associates an $N$-qubit operator 
$G_i$ defined as
\begin{equation} \label{Gi}
    G_{i} = X_{i}\otimes  \bigotimes_{j \in \mathcal{N}(i)} Z_{j}, 
\end{equation} 
where the single $X$ acts on site $i$, whereas the $Z$ operators act on all sites that belong to the neighbourhood $\mathcal{N}(i)$ of $i$. Having introduced the $G_i$ operators we define the graph state.
\begin{defi}\label{def:gs}
We define the graph state $|G \rangle $ associated to the graph $\mathcal{G}=(\mathcal{V},\mathcal{E})$ as the unique state stabilized by the corresponding operators $G_i$ \eqref{Gi}, that is,
\begin{equation}\label{eq2}
   G_{i}|G \rangle = |G \rangle ,\qquad  \forall i=1,\dots,N.
\end{equation}
In other words, $\ket{G}$ is the unique common eigenstate of all $G_i$ corresponding to eigenvalue $+1$.
\end{defi}
The operators $G_i$ are usually called the stabilizing operators. Notice also that they mutually commute and the Abelian group generated by them is called a stabilizer.

Two simple examples of connected graphs with three vertices are depicted in Fig. \ref{graphs1and2}. The graph on the left is a complete graph, i.e., one in which any vertex is connected to any other vertex by an edge. The unique three-qubit state associated to this graph is stabilized by the following three stabilizing operators:
\begin{eqnarray}\label{Stab3}
G_1 &=& X \otimes Z \otimes Z, \\
G_2 &=& Z \otimes X \otimes Z, \\
G_3 &=& Z \otimes Z \otimes X,
\end{eqnarray}
and can be stated as
\begin{eqnarray}\label{FCg}
|G' \rangle  &=& \frac{1}{\sqrt{8}}( 
{\left| 000 \right\rangle} 
+ {\left| 100 \right\rangle} 
+ {\left| 010 \right\rangle} 
- {\left| 110 \right\rangle} \nonumber \\
&&\hspace{0.75cm}+ {\left| 001 \right\rangle} 
- {\left| 101 \right\rangle} 
- {\left| 011 \right\rangle} 
- {\left| 111 \right\rangle} 
).
\end{eqnarray}
The graph on the right side in Fig. \ref{graphs1and2} is a non-isomorphic to the complete graph. The unique three-qubit state associated with this graph is stabilized by  
\begin{eqnarray}
G_1 &=& X \otimes Z \otimes Z, \\
G_2 &=& Z \otimes X \otimes \mathbbm{1}, \\
G_3 &=& Z \otimes \mathbbm{1} \otimes X,
\end{eqnarray}
and is given by
\begin{eqnarray}\label{nFCg}
|G'' \rangle  &=&
 \frac{1}{\sqrt{8}}( 
{\left| 000 \right\rangle} 
+ {\left| 100 \right\rangle} 
+ {\left| 010 \right\rangle} 
- {\left| 110 \right\rangle} \nonumber \\
&&\hspace{0.75cm}+ {\left| 001 \right\rangle} 
+ {\left| 101 \right\rangle} 
- {\left| 011 \right\rangle} 
+ {\left| 111 \right\rangle}
).
\end{eqnarray}

\begin{figure}
    \centering
    \includegraphics[scale=0.26]{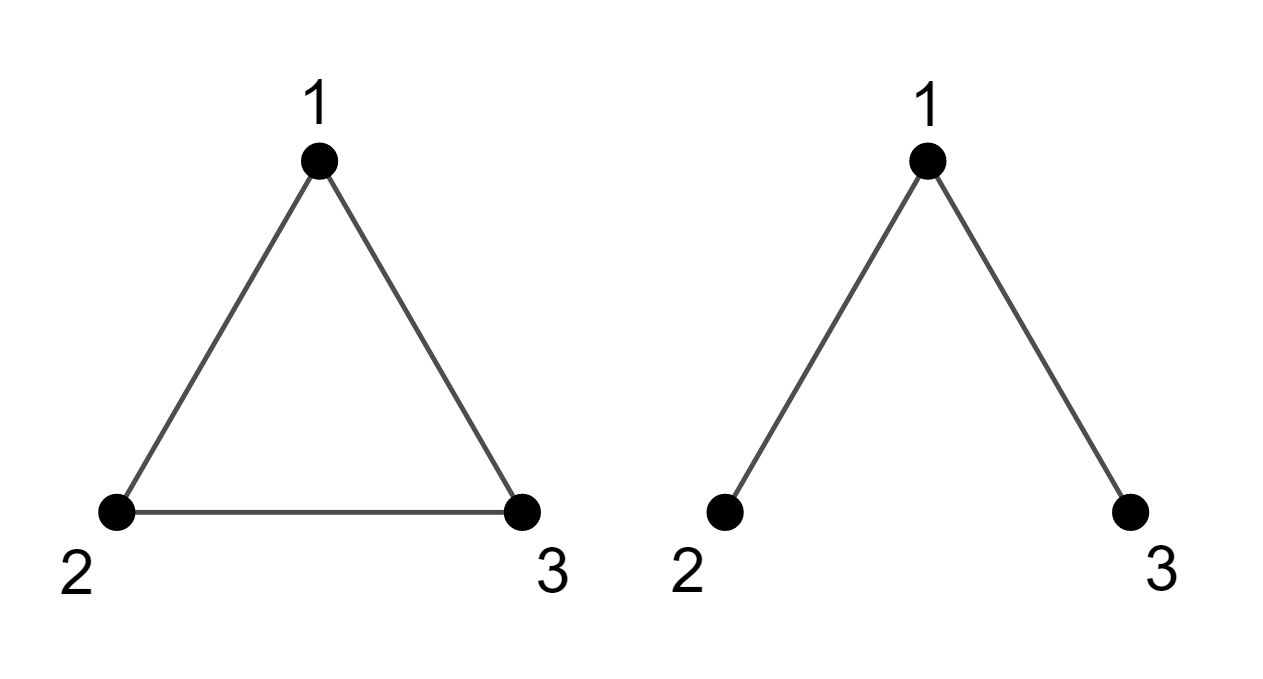}
    \caption{Two non-isomorphic graphs with 3 vertices.} \label{graphs1and2}
  \end{figure}

Let us notice that although both these exemplary states $\ket{G'}$
and $\ket{G''}$ correspond to non-isomorphic graphs, they are actually 
equivalent to the same three-qubit Greenberger-Horne-Zeilinger (GHZ) state, 
\begin{equation}\label{GHZstate}
    \ket{\mathrm{GHZ}}=\frac{1}{\sqrt{2}}(\ket{000}-\ket{111}),
\end{equation}
by suitable local unitary transformations.

In the context of multipartite qubit states such vertex is associated to a qubit and edges represent entanglement between qubits. 
However, in our scheme that we propose for certification of such multiqubit state, we do not assume that there exists a local Hilbert space corresponding to each vertex.

Let us finally notice that the construction of the graph state
corresponding to the three-vertex complete graph can be generalized
to any number of qubits. The corresponding stabilizing operators
are given by 
\begin{equation}\label{ngraph}
    G_i=Z_1\ldots Z_{i-1}X_iZ_{i+1}\ldots Z_n,
\end{equation}
with $i=1,\ldots,n$. They stabilize an $n$-qubit graph state 
that is local-unitary equivalent to the GHZ state $(1/\sqrt{2})(\ket{0}^{\otimes n}+\ket{1}^{\otimes n})$.

\subsection{Self-testing}

Self-testing, originally defined in Ref. \cite{MY04} in the context of non-locality, aims to certify an unknown quantum state and a set of measurements based on statistics obtained in an experiment, up to certain unitary equivalence and the existence of auxiliary degrees of freedom. Self-testing based on violation of Bell inequalities is by definition a device-independent task as it does not depend on any assumption on the state and measurements. In a Bell test, the assumption of commutativity between the measurements arises due to the fact that spatially separated subsystems cannot communicate instantaneously with each other.

Self-testing statements based on violation of noncontextuality inequalities require, on the other hand, the assumption of compatibility of measurements. First, contextuality-based self-testing was defined
in Ref. \cite{BRV+19} in a similar way to how
self-testing is defined within the Bell scenario. Here 
we provide a slightly different definition that 
takes inspiration from Ref. \cite{IMO+20} 
(see also Ref. \cite{SSA20})
and fits better the inequalities introduced here.

To this aim, let us consider again the experiment described in 
Sec. \ref{scenario}, but now we assume that both the state (in general mixed) and measurements (in general non-projective)
are unknown; still, the measurements obey certain 
compatibility relations. Since 
%
%the violation of a noncontextuality inequality
%in the device-independent context, i.e., by using an unknown quantum state and a 
%set of unknown dichotomic measurements which are a positive operator valued  measure %(POVM), but satisfying the relevant compatibility
%relations. 
%
we do not specify the dimension of the underlying Hilbert space,
without loss of generality, we can assume that
the measurements are projective and the state is pure  (see, e.g., Ref. \cite{IMO+20}). 
In other words, any correlations giving rise to the violation of the noncontexuality inequality can always be reproduced with a pure state 
$\rho=|\psi\rangle\!\langle\psi|$ and observables $A_i$ obeying $A_i^2=\mathbbm{1}$,
all acting on some Hilbert space of unknown dimension $\mathcal{H}$.
These observables obey the same compatibility relations.

At the same time we consider a reference experiment 
with a known pure state $\ket{\tilde{\psi}}\in\mathbbm{C}^d$ for some
$d$ and known observables $\tilde{A}_i$ acting on $\mathbbm{C}^d$
that obey the same compatibility relations.
\begin{defi}\label{defselftesting}
Suppose an unknown  state $|\psi \ra  \in \mathcal{H}$ and a set of measurements $ A_i $   violate a given noncontextuality inequality maximally, then this maximal quantum violation 
self-tests the state $| \tilde{\psi}  \ra  \in \mathbbm{C}^d$ and the set of measurements $ \tilde{A_i} $ if there exists exists a projection $P:\mathcal{H} \rightarrow \mathbb{C}^d$ and a unitary $U$ acting on $\mathbbm{C}^d$ such that
\beq\label{DefSelf}
U^\dagger (P\, A_i\, P^\dagger) U &=& \tilde{A}_i, \\
U (P\,|\psi \ra) &=& |\tilde{\psi} \ra.
\eeq
\end{defi}
Speaking alternatively, the above definition says 
that based on the observed nonclassicality one is able to identify 
a subspace $V=\mathbbm{C}^d$ in $\mathcal{H}$ on which 
all the observables act invariantly. Equivalently, 
$A_i$ can be decomposed as $A_i=\hat{A}_i\oplus A_{i}'$, where
$\hat{A}_i$ act on $V$, whereas $A_i'$ act on the orthocomplement of $V$ in $\mathcal{H}$; in particular, $A_i'\ket{\psi}=0$. Moreover, there is a unitary 
$U^{\dagger}\,\hat{A}_i\,U=\tilde{A}_i$.

\section{The simplest inequality and Self-testing of three-qubit graph state}\label{simplestinequality}

\subsection{Simplest inequality}

Here, we consider a noncontextuality inequality that allows one for self-testing the complete graph state of three qubits and simultaneously a set of six dichotomic observables denoted by $A_i$ and $B_j$ $(i,j\in\{ 1,2,3 \})$ such that $\{A_i,B_i\}=0$ ($i=1,2,3$).
%Here we assume that all measurements commute, except within $A_i$ and $B_i$ for $i=\{ 1,2,3 \}$.
The measurements outcomes are labelled by $\pm1$, which means that the measurement operators have eigenvalues $\pm1$ and thus they satisfy $A^2_i=B^2_i=\I$. 
%Moreover these measurements are performed in certain triples within which they 
%are compatible. 

The compatibility hypergraph of the scenario is depicted in Figure \ref{fig1}. A compatibility hypergraph is one in which the vertices are associated with the measurements of the scenario and the hyperedges represent the contexts which are subsets of compatible measurements. The noncontextuality inequality we consider is given by
\begin{align} \label{Ineq3}
\mathcal{I}_3 & :=  \la A_1B_2B_3 \ra + \la B_1A_2B_3 \ra + \la B_1B_2A_3 \ra - \la A_1A_2A_3 \ra \nonumber \\ 
&  \le  \eta_C = 2 < \eta_Q = 4.
\end{align}
The above inequality is equivalent to a noncontextuality inequality employed in Ref. \cite{CEG+14} to demonstrate quantum contextuality of a single eight-dimensional quantum system. In the context of Bell scenario it is the well-known Mermin-Ardehali-Belinskii-Klyshko (MABK) inequality \cite{PhysRevA.46.5375,Uspekhi,Mer90}
for which a non-locality-based self-testing statement was derived in Ref. \cite{Kan16}.

Following the argument in the previous section, the classical bound of the above expression can be obtained by assigning the values $\pm1$ to each variable $A_i$ and $B_j$ which implies $\eta_C = 2$. At the same time, the algebraic maximum of $\mathcal{I}_3$ is four since the correlators can take a value between $-1$ and $+1$, and, importantly, it equals the maximal quantum value of $\mathcal{I}_3$, that is $\eta_Q=4$. Indeed, it can be checked that the following set of measurements:
\beq \label{AiBj}
A_1 &=& X \otimes \mathbbm{1} \otimes \mathbbm{1}, \qquad B_1 = Z \otimes \mathbbm{1} \otimes \mathbbm{1}, \nonumber \\
A_2 &=& \mathbbm{1} \otimes X \otimes \mathbbm{1}, \qquad
B_2 = \mathbbm{1} \otimes Z \otimes \mathbbm{1}, \nonumber \\
A_3 &=&  \mathbbm{1} \otimes \mathbbm{1} \otimes X, \qquad
B_3 = \mathbbm{1} \otimes \mathbbm{1} \otimes Z,
\eeq
together with the complete graph state $|G' \ra$ given by Eq. (\ref{FCg}) give rise to the algebraic maximum. This follows from the fact that for this choice of observables, the first three terms of the inequality correspond to 
the stabilizing operators $G_i$ given in Eq. (\ref{Stab3}), 
whereas the last one to their product $G_1G_2G_3=-X_1X_2X_3$.

Let us notice that almost all pairs of these observables commute except for
those with the same substripts which anticommute, that is,
\begin{equation}\label{dupa}
    [A_i,A_j]=[A_i,B_j]=[B_i,B_j]=0\qquad (i\neq j)
\end{equation}
and
\begin{equation}
    \{A_i,B_i\}=0\qquad (i=1,2,3).
\end{equation}

For further convenience let us also comment on the symmetries of the inequality (\ref{Ineq3}). A simple way to visualize these symmetries is by using the compatibility hypergraph of the scenario that is represented in Fig. \ref{fig1}. For instance, note that under the relabeling of the measurements $A_i \leftrightarrow A_j$ together with $B_i \leftrightarrow B_j$, i.e., a permutation of subscripts $i \leftrightarrow j$, the inequality remains the same. We can also observe that a cyclic permutation of measurements $A_1 \rightarrow A_2 \rightarrow A_3 \rightarrow A_1$ together with $B_3 \rightarrow B_1 \rightarrow B_2 \rightarrow B_3$ does not change the structure of the hypergraph and therefore the inequality remains the same. 
Other symmetries can be found just by looking at the hypergraph since it captures the intrinsic structure of the associated inequality. The above symmetries as well as the structure of the $\mathcal{I}_3$ expression will be vital for our considerations, in particular, for generalizing this inequality into a family
of inequalities maximally violated by $n$-qubit GHZ states and $n$ pairs of
anticommuting observables.

\begin{figure}
    \centering
    \includegraphics[scale=0.5]{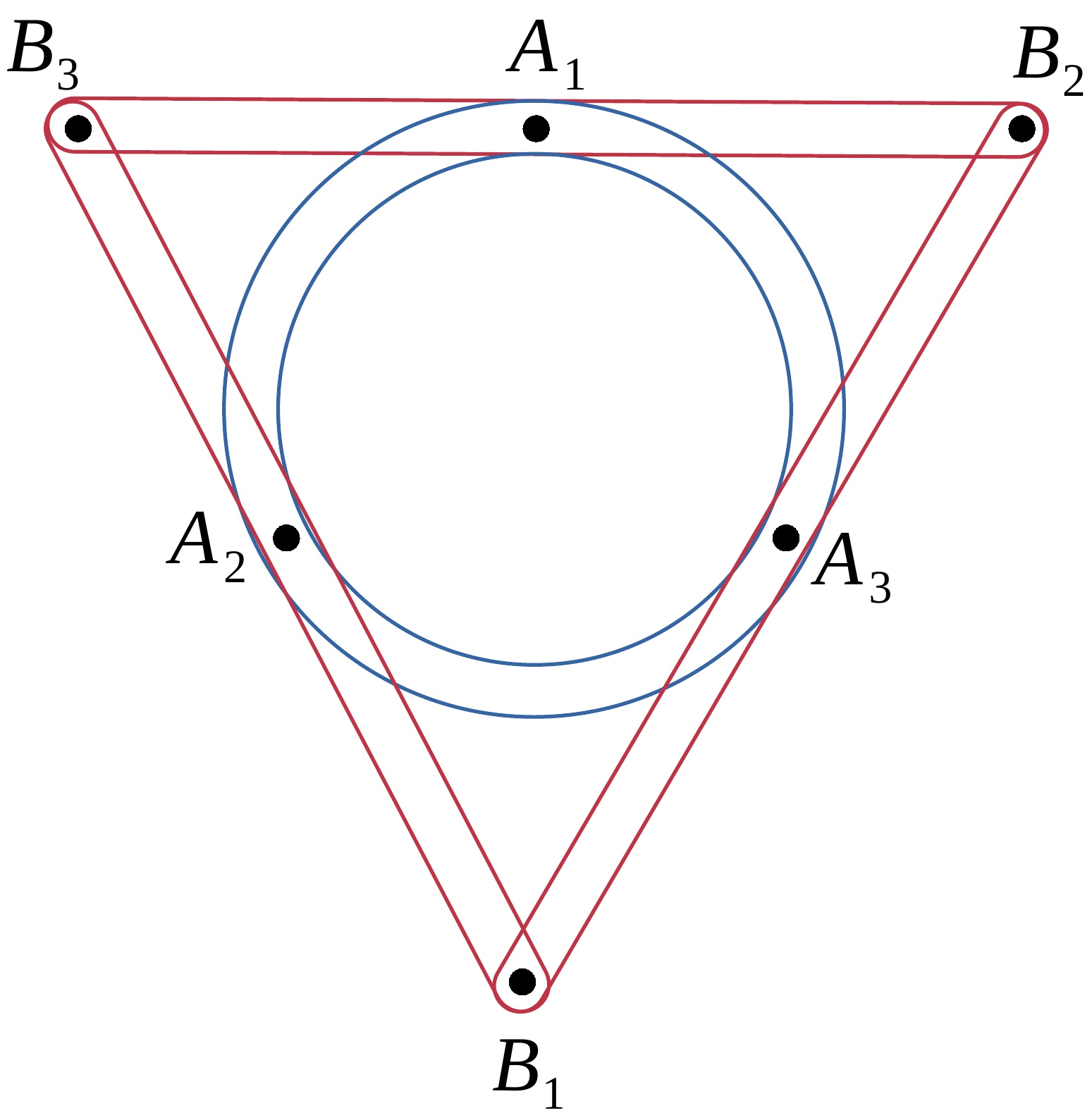}
    \caption{Hypergraph \cite{Bretto-MR3077516} of compatibility of the Kochen-Specker contextuality scenario  associated to the inequality \eqref{Ineq3}. In this hypergraph, the vertices represent the observables of the scenario and the hyperedges represent the contexts. The red hyperedges are associated to the correlators which enter the inequality \eqref{Ineq3} with $+1$ and the blue to correlators corresponding to the negative sign. Here the colors are conveniently chosen in order to elucidate symmetric properties of the inequality.} \label{fig1}
 \end{figure}

\subsection{Self-testing}

We now prove that the maximal quantum violation of the inequality (\ref{Ineq3}) can be used for certification of the GHZ state (\ref{FCg}) along with the observables (\ref{AiBj}). To this aim, consider a quantum realization given by a pure state $|\psi\ra\in\mathcal{H}$ and a set of quantum observables $A_i,B_j$ with $i,j=1,2,3$ acting on $\mathcal{H}$, where $\mathcal{H}$ is some unknown Hilbert space. We additionally assume that these observables obey the compatibility relations presented in Fig. \ref{fig1}, which translate into the 
commutation relations in Eq. (\ref{dupa}).

Assume then the correlations obtained by measuring $A_i$ and $B_j$ on the state $\ket{\psi}$ attains the quantum bound of the inequality \eqref{Ineq3}. 
This directly implies that the first three terms in $\mathcal{I}_3$
take value $1$, whereas the last term equals $-1$, which
via the Cauchy-Schwartz inequality translate to the following equations:
\beq\label{qrealization1}
A_1B_2B_3 |\psi\ra &=& |\psi\ra \quad \& \quad \texttt{permutations}, \\ \label{qrealization2}
B_1A_2B_3 |\psi\ra &=& |\psi\ra  \quad \& \quad \texttt{permutations},   \\ \label{qrealization3}
B_1B_2A_3 |\psi\ra &=& |\psi\ra  \quad \& \quad \texttt{permutations},    \\ \label{qrealizations4} 
A_1A_2A_3 |\psi\ra &=& -|\psi\ra  \quad \& \quad \texttt{permutations},
\eeq
where \texttt{permutations} refers to the fact that the above relations also hold if we permute the observables, which is a consequence of the commutation relations (\ref{dupa}). One directly deduces from these identities that
\beq \label{preSym}
A_1B_1 |\psi\ra &=& A_1A_2B_3 |\psi\ra = -B_3A_3 |\psi\ra \nonumber \\
                &=& A_1A_3B_2 |\psi\ra = A_3B_3 |\psi\ra, 
\eeq
where in the first line we used $B_1 |\psi\ra = A_2B_3 |\psi\ra$ from Eq. \eqref{qrealization2} and then the fact that $B_3$ commutes with $A_1$ and $A_2$ along with the relation $A_1A_2 |\psi\ra = -A_3|\psi\ra$ that stems from Eq. \eqref{qrealizations4}. On the other hand, in the second line, we used $B_1 |\psi\ra = A_3B_2 |\psi\ra$ from Eq. \eqref{qrealization3} and $A_1B_2 |\psi\ra = A_3|\psi\ra$ from Eq. \eqref{qrealizations4}. 

Let us now employ the symmetries of the inequality. Indeed, as already mentioned, it is invariant under simultaneous permutations $A_i \leftrightarrow A_j$ and $B_i \leftrightarrow B_j$ for any $i\neq j$, and therefore one can straightforwardly infer from Eq. (\ref{preSym}) that the following chain of equalities hold true,
\beq \label{symm1}
A_1B_1 |\psi\ra &=& A_2B_2 |\psi\ra = A_3B_3 |\psi\ra \nonumber \\ &=& -B_1A_1 |\psi\ra = -B_2A_2 |\psi\ra = -B_3A_3 |\psi\ra.\nonumber\\
\eeq
From the above equations it follows that the operators $A_i$ and $B_i$ with $i=1,2,3$ anticommute when acting on the state
$|\psi\ra $, i.e.,
\beq\label{anticom}
\{A_1,B_1\}|\psi\ra = \{A_2,B_2\}|\psi\ra = \{A_3,B_3\}|\psi\ra = 0.
\eeq

Inspired by the approach of Ref. \cite{IMO+20}, we now define a subspace 
\beq \label{defInv}
    V &:=& \mathrm{span} \{ |\psi\ra , A_1 |\psi\ra, A_2 |\psi\ra, A_3 |\psi\ra, \nonumber \\ &&\hspace{1cm}B_1 |\psi\ra, B_2 |\psi\ra, B_3 |\psi\ra, A_1 B_1 |\psi\ra \},
\eeq
and prove the following fact for it.
\begin{lemma}\label{invsubV}
$V$ is an invariant subspace of all the observables $A_i$ and $B_j$ for $i,j\in\{ 1,2,3 \}$.
\end{lemma}
\begin{proof}
One can verify with the aid of Eqs. (\ref{qrealization1})-(\ref{qrealizations4}) that the action of the operator $A_1$ on the eight vectors spanning $V$ is a permutation of these vectors up to the factor $-1$. In exactly the same way one shows that $B_1V\subseteq V$. Therefore, we conclude that the subspace $V$ is invariant under the action of the observables $A_1$ and $B_1$. By the symmetry of the inequality (\ref{Ineq3}) it then follows that the subspace $V$ is invariant under the action of all observables $A_i$ and $B_j$ for $i,j\in\{ 1,2,3 \}$.
\end{proof}

It should be noticed that due to  Eq. \eqref{symm1}, the subspace $V$ stays the same if one replaces the last vector $A_1B_1\ket{\psi}$ in Eq. (\ref{defInv}) by $A_2B_2 |\psi\ra$ or $A_3B_3 |\psi\ra$. 

Due to Lemma \ref{invsubV}, it suffices for our purpose to identify the form of the state $|\psi\ra$ and the operators $A_i$ and $B_j$ restricted to the subspace $V$. 
In fact, the whole Hilbert space splits as $\mathcal{H}=V\oplus V^{\perp}$, 
where $V^{\perp}$ is an orthocomplement of $V$ in $\mathcal{H}$. 
Then, the fact that $V$ is an invariant subspace of all
the observables $A_i$ and $B_j$ means that they have the 
following block structure 
\begin{equation}\label{block}
    A_i=\hat{A}_i\oplus A_i',\qquad B_{j}=\hat{B}_j\oplus B_j',
\end{equation}
where $\hat{A}_i=PA_iP$ and analogously $\hat{B}_i=PB_i P$ with $P:\mathcal{H}\to V$ being a projection onto $V$. Since $A_i'$ and $B_j'$ 
act trivially on $V$, that is $A_i'V=B_j'V=0$, which means that the
observed correlations giving rise to the maximal violation of 
the inequality (\ref{Ineq3}) come solely from the subspace $V$, 
in what follows we can restrict our attention to the 
operators $\hat{A}_i$ and $\hat{B}_j$. 

First, from the fact that $A_i$ and $B_j$ are observables obeying $A_i^2=B_j^2=\mathbbm{1}$, it directly follows that $\hat{A}_i$, $\hat{B}_j$ are observables too and satisfy 
\begin{equation}\label{dupa3}
 \hat{A}_i^2=\hat{B}_j^2=\mathbbm{1}_V\qquad (i,j=1,2,3),   
\end{equation}
where $\mathbbm{1}_V$ is the identity acting on $V$.
Second, Eq. (\ref{block}) implies that the hatted observables must obey the same commutation relations as $A_i$ and $B_j$, that is, 
\begin{equation}
    [\hat{A}_i,\hat{A}_j]=[\hat{A}_i,\hat{B}_j]=[\hat{B}_i,\hat{B}_j]=0\qquad (i\neq j)
\end{equation}
Third, it turns out that the relations (\ref{anticom}) force $\hat{A}_i$
and $\hat{B}_i$ to anticommute on the subspace $V$. 

\begin{lemma}\label{lemma2}
Suppose the maximal quantum violation of the inequality (\ref{Ineq3}) is observed. 
Then, $\{ \hat{A}_i,\hat{B}_i \} = 0$ for all $i\in\{ 1,2,3 \}$.
\end{lemma}
\begin{proof}
Let us focus on the first pair $\hat{A}_1$ and $\hat{B}_1$. 
By checking the action of $\{ A_1,B_1 \}$ on all the eight vectors that span the subspace $V$ we can conclude that $\{ \hat{A}_1,\hat{B}_1 \}=0$. Indeed, Eq. (\ref{anticom}) implies that $\{ A_1,B_1 \}$ vanishes when acting on $\ket{\psi}$. Then, for $A_1\ket{\psi}$ one has 
\begin{equation}
 \{ A_1,B_1 \}A_1\ket{\psi}=(A_1B_1A_1+B_1)\ket{\psi}=0,   
\end{equation}
where the second equality follows again from Eq. (\ref{anticom}).
For $A_2\ket{\psi}$ and $A_3\ket{\psi}$ we can use the fact that 
$A_1$ and $B_1$ commute with $A_2$ and $A_3$, which gives 
\begin{equation}
    \{A_1,B_1\}A_i\ket{\psi}=A_i\{A_1,B_1\}\ket{\psi}=0\qquad (i=2,3).
\end{equation}
In exactly the same way one deals with the vectors $B_i\ket{\psi}$.
Finally, for the last vector in (\ref{defInv}), $A_1B_1\ket{\psi}$, 
one has
\begin{equation}
    \{A_1,B_1\}A_1B_1\ket{\psi}=(A_1B_1A_1B_1+\mathbbm{1})\ket{\psi}=0,
\end{equation}
where to get the last equality we use Eq. (\ref{anticom}). Owing
to the block form of $A_1$ and $B_1$ in Eq. (\ref{block}), all this implies 
that $\{\hat{A}_1,\hat{B}_1\}=0$.

One more time, by the symmetries of the inequality, we can draw the same conclusions for the remaining pairs of the observables $A_i$ and $B_i$. 
As a result $\{\hat{A}_i,\hat{B}_i\}=0$ for $i=1,2,3$, which completes the proof.
\end{proof}

With Lemma \ref{lemma2} at hand, we can now employ the standard
approach that has already been used in many non-locality-based self-testing 
schemes \cite{Kan16,KST+19,SSK+19,BAS+20}. Precisely, using this approach we can first 
infer that the dimension $d$ of the subspace $V$ is even. 
To see this, note that from the above anticommutation 
relation between $\hat{A}_i$ and $\hat{B}_i$, we have
\beq
\hat{A}_i = - \hat{B}_i\hat{A}_i\hat{B}_i\quad\mathrm{or}\quad
\hat{B}_i = - \hat{A}_i\hat{B}_i\hat{A}_i,
\eeq
which after taking trace on both sides simplify to $\tr(\hat{A}_i)=\tr(\hat{B}_i)=0$. It then follows that both the eigenvalues $\pm1$ of each observable $\hat{A}_i$ or $\hat{B}_i$ have equal multiplicities. This clearly implies that 
the dimension $d=\dim V$ is an even number, $d=2k$ for some $k\in\mathbbm{N}$, and thus $V=\mathbbm{C}^2\otimes\mathbbm{C}^k$. On the other hand, since $\dim V\leq 8$,
one concludes that $k=2,3,4$.

The fact that $\hat{A}_1$ and $\hat{B}_1$ are traceless means also that the operators $\hat{A}_1$ and $\hat{B}_1$ are equivalent to $X \otimes \mathbbm{1}_k$ and $Z \otimes \mathbbm{1}_k$ for some $k=2,3,4$ up to some unitary operation (see for instance Appendix B in Ref. \cite{KST+19} for the proof of this statement). This observation is one of the key ideas behind the proof of the following lemma. 

\begin{lemma}\label{lemma3}
Suppose the maximal quantum violation of the inequality (\ref{Ineq3}) is observed. Then, there exists a basis in $V$ such that 
\beq \label{AiBjtilde}
\hat{A}_1 &=& X \otimes \mathbbm{1} \otimes \mathbbm{1}, \qquad
\hat{B}_1 = Z \otimes \mathbbm{1} \otimes \mathbbm{1}, \nonumber \\
\hat{A}_2 &=& \mathbbm{1} \otimes X \otimes \mathbbm{1},\qquad
\hat{B}_2 = \mathbbm{1} \otimes Z \otimes \mathbbm{1}, \nonumber \\
\hat{A}_3 &=&  \mathbbm{1} \otimes \mathbbm{1} \otimes X, \qquad 
\hat{B}_3 = \mathbbm{1} \otimes \mathbbm{1} \otimes Z.
\eeq
\end{lemma}
\begin{proof}
First, from Lemma \ref{lemma2} we have $\{ \hat{A}_1,\hat{B}_1 \} = 0$ which implies that there exists a unitary $U_1$ acting on $V$ such that
\begin{eqnarray}\label{IdA1B1}
U_1^{\dagger}\, \hat{A}_1\, U_1 &=& X \otimes \mathbbm{1}_k, \\
U_1^{\dagger}\, \hat{B}_1\, U_1&=& Z \otimes \mathbbm{1}_k,
\end{eqnarray}
where, as already mentioned, the dimension $d$ of the subspace $V$ is given by $d=2k$ for some $k=2,3,4$. Using then the above form of $\hat{A}_1$ and $\hat{B}_1$ 
and the commutation relations (\ref{dupa}) we can write the other operators 
as follows:
\beq
\label{form1}U^{\dag}_1\,\hat{A}_2\,U_1 &=& \mathbbm{1}_2 \otimes M, \\
U^{\dag}_1\,\hat{B}_2\,U_1 &=& \mathbbm{1}_2 \otimes N, \\
U^{\dag}_1\,\hat{A}_3\,U_1 &=&  \mathbbm{1}_2 \otimes O, \\
U^{\dag}_1\,\hat{B}_3\,U_1 &=& \mathbbm{1}_2 \otimes P,
\eeq
where $M,N,O,P$ are Hermitian involutions acting on the subspace of dimension $k$. 
To show explicitly how the above equations are obtained let us focus on $\hat{A}_2$;
the proof for the other observables is basically the same. Since $\hat{A}_2$ acts on
$\mathbbm{C}^2\otimes\mathbbm{C}^k$, it can be decomposed in the Pauli basis as
\begin{equation}
    U_1^{\dagger}\,\hat{A}_2\,U_1=\mathbbm{1}_2\otimes M_1+X\otimes M_2+Y\otimes M_3+Z\otimes M_4,
\end{equation}
where $Y$ is the third Pauli matrix and $M_i$ are some Hermitian matrices acting on $\mathbbm{C}^k$. Now, it follows from the fact that $\hat{A}_2$ commutes with $\hat{A}_1$, that $M_3=M_4=0$. Then, from $[\hat{A}_2,\hat{B}_1]$ one obtains that
$M_2=0$, and, by putting $M_1=M$, we arrive at Eq. (\ref{form1}).

Second, from Lemma \eqref{lemma2}, we  have $\{ \hat{A}_2,\hat{B}_2\} = 0$ which is equivalent to $\{M,N\}=0$. Since both $M$ and $N$ are involutions, one concludes, as before, that $k=2k'$ for $k'=1,2$, or, equivalently, that $\mathbbm{C}^k=\mathbbm{C}^2\otimes\mathbbm{C}^{k'}$. Moreover, there exists another unitary transformation $U_2:\mathbbm{C}^k\to\mathbbm{C}^k$ such that
\begin{eqnarray}
    U_2^{\dagger}\,M\, U_2&=&X\otimes\mathbbm{1}_{k'},\\
    U_2^{\dagger}\,N\, U_2&=&Z\otimes\mathbbm{1}_{k'}.
\end{eqnarray}
Finally, to learn the form of $O$ and $P$ we can again employ the commutation
relations (\ref{dupa}). They imply in particular that $[M,O]=[N,O]=[M,P]=[N,P]=0$, 
and consequently,
\begin{equation}
    O=\mathbbm{1}_2\otimes O',\qquad P=\mathbbm{1}_2\otimes P',
\end{equation}
where $O'$ and $P'$ are some operators acting on $\mathbbm{C}^{k'}$ such that 
$[O']^2=[P']^2=\mathbbm{1}_{k'}$. Additionally, due to the fact that $\{\hat{A}_3,\hat{B}_3\}=0$, they must anticommute, $\{O',P'\}=0$.
This means that $k'=2$ and that there exist a unitary operation $U_3$
acting on this qubit Hilbert space such that 
\begin{equation}
    U_3^{\dagger}\,O'\,U_3=X,\qquad U_3^{\dagger}\,P'\,U_3=Z.
\end{equation}

Taking all this into account, one finds that $V\cong\mathbbm{C}^2\otimes\mathbbm{C}^2\otimes\mathbbm{C}^2$
and that there exists a single unitary operation $U = U_1 (\mathbbm{1}_2 \otimes U_2) (\mathbbm{1}_2 \otimes \mathbbm{1}_2 \otimes U_3)$ on $V$ which brings
all the observables $\hat{A}_i$ and $\hat{B}_i$ to the form in
(\ref{AiBjtilde}).
\end{proof}

We have thus arrived at one of our main results of this paper.

\begin{thm}\label{Theo3qubit}
If a quantum state $|\psi \ra$ and a set of measurements $A_i$ and $B_j$ with $i,j\in\{ 1,2,3 \}$ maximally violate the inequality \eqref{Ineq3}, then there exist a projection $P:\mathcal{H} \rightarrow V$ with $V=(\mathbbm{C}^2)^{\otimes 3}$ and a unitary $U$ acting on $V$ such that
\beq\label{dupawolowa}
U^\dagger\, (P\, A_i\, P^\dagger)\, U &=& X_i, \\
U^\dagger\, (P\, B_i\, P^\dagger)\, U &=& Z_i, 
\eeq
%
%\beq
%U^\dagger (P A_i P^\dagger) U &=& \hat{A}_i, \\
%U^\dagger (P B_j P^\dagger) U &=& \hat{B}_j \\
%U (P|\psi \ra) &=& |G \ra
%\eeq
%
where $X_i$ and $Z_i$ are $X$ and $Z$ Pauli matrices acting on qubit $i$,
and
\begin{equation}
    U (P|\psi \ra) = |G \ra
\end{equation}
with $|G \ra$ being the three-qubit complete graph state defined in (\ref{FCg}).
\end{thm}
\begin{proof}
A quantum state $\ket{\psi}$ that belongs to a Hilbert space $\mathcal{H}$ and a set of observables $A_i,B_j$ acting on $\mathcal{H}$ attain the maximal quantum violation of the inequality \eqref{Ineq3} if and only if they satisfy the set of Eqs. \eqref{qrealization1}-\eqref{qrealizations4}. The algebraic relations induced by this set of equations let us prove Lemmas \ref{invsubV}-\ref{lemma3} which imply that there exists a projection $P:\mathcal{H} \rightarrow V \cong \mathbbm{C}^8 $ and a unitary $U = U_1 (\mathbbm{1}_2 \otimes U_2) (\mathbbm{1}_2 \otimes \mathbbm{1}_2 \otimes U_3)$ acting on $V \cong \mathbbm{C}^8$ for which Eqs. (\ref{dupawolowa})
hold true.

From the above characterization of the observables we can infer the form of the state $\ket{\psi}$. Indeed, after plugging (\ref{Aitilden}) into the conditions (\ref{qrealization1}) one realizes that the latter are simply the stabilizing conditions of the graph state associated to the complete graph of three vertices given in Eq. (\ref{FCg}) and thus $U\ket{\psi}=\ket{G'}$. This completes the proof.
\end{proof}

%\subsection{A semi device-independent approach}\label{sdia}

A few comments are in order. The first is that the tensor product structure here is just a suitable mathematical tool we used to represent our results. We know that in composite quantum systems the Hilbert space of the whole system is a tensor product of the Hilbert spaces of the separate subsystems, however, it has to be clear that in Theorem \ref{Theo3qubit} we do not assume the whole system to be composite.

The second comment is that the certification statement made in Theorem 
\ref{Theo3qubit} involves a global unitary operation which means that 
any state from $V$ can in fact be chosen as the reference state, even 
a fully product one. Thus, Theorem (\ref{Theo3qubit}) cannot be understood as a certification of only the state, but rather as certification of a state and a set of measurements at the same time. Or, more precisely, it is a certification of how measurements act on the state or what the relation between a state and measurements is; this relation is basis independent.

One way to get rid of the above ambiguity is to assume that the quantum system
at hand is composed of spatially separated subsystems on which the verifier
can perform local measurements. Such an assumption allows 
them to use Bell non-locality to deduce the form of the state. For instance
for the GHZ state of three-qubit a self-testing statement based on the violation of the inequality (\ref{Ineq3}) was derived in Ref. \cite{Kan16}.

To illustrate the difference between contextuality and non-locality based certification let us consider another set of quantum observables
on $\mathbbm{C}^8$ defined as
\begin{eqnarray}\label{anotherset}
 A_1&=&X \otimes \mathbbm{1} \otimes \mathbbm{1},\quad A_2=\mathbbm{1} \otimes X \otimes Z,\quad A_3=\mathbbm{1} \otimes Z \otimes X,\nonumber\\ 
 B_1&=&Z \otimes \mathbbm{1} \otimes \mathbbm{1},\quad\, B_2=\mathbbm{1} \otimes Z \otimes \mathbbm{1},\hspace{0.1cm}\,\quad B_3=\mathbbm{1} \otimes \mathbbm{1} \otimes Z.\nonumber\\
\end{eqnarray}
Clearly, these observables, similarly to those in Eq. (\ref{AiBj}), satisfy the commutation and anticommutation relations in Eqs. (\ref{dupa}) and (\ref{anticom}).
Moreover, they give rise to the maximal violation of the inequality (\ref{Ineq3})
together with the graph state $\ket{G''}$ corresponding to the linear graph in Fig. (\ref{graphs1and2}). However, while both the graph states $\ket{G'}$ and $\ket{G''}$
are equivalent under local unitary operations and thus cannot be distinguished
within both approaches to self-testing, the sets of observables
in Eqs. (\ref{AiBj}) and (\ref{anotherset}) are certainly not; they are equivalent under global unitary operations. Thus, the maximal violation of the Bell inequality (\ref{Ineq3}) would allow one to distinguish between these two sets, while standard quantum contextuality does not allow to do that.

\section{A scalable inequality and Self-testing of multi-qubit graph states}\label{scalable}

In this section we design a family of noncontextuality inequalities which is scalable and aimed to certify multiqubit quantum systems. These inequalities are scalable since the number of measurements and correlators increase polynomially with the number of vertices of the respective graph state.  The inequality we propose in \eqref{Ineqn} generalizes the inequality given in \eqref{Ineq3} and has this inequality in the heart of the construction since the structure of the simplest inequality appears as the building blocks of the general construction. We prove that the inequalities are useful for certification purposes.

\subsection{Scalable noncontextuality inequalities}

First, let us consider a set of $2n$ observables denoted by $A_1,\dots,A_n$ and $B_1,\dots,B_n$. They are assumed to mutually commute except for pairs $A_i,B_i$ with $i \in \{ 1,\ldots,n \}$, that is,
\begin{equation}\label{commutn}
 [A_i,A_j]=[B_i,B_j]=[A_i,B_j]=0\qquad (i\neq j).
\end{equation}
We now describe our construction of the noncontextuality inequalities. We first consider a sum of $n$ expectation values of the form $C_i=\la B_1\ldots B_{i-1} A_i B_{i+1}\ldots B_n \ra $ for $i=1,\ldots,n$ which involve $n-1$ different $B_j$ observables and a single observable $A_i$. Then, for any choice of three out of $n$ such different correlators $C_i$, $C_j$ and $C_k$ $(i\neq j\neq k)$
we consider another correlator that we substract from the sum. It is given by
\begin{equation}
    \la B_1\ldots A_i\ldots A_j\ldots A_k\ldots B_n \ra
\end{equation}
and consists of three observables $A_i$, $A_j$ and $A_k$ and 
$n-3$ observables $B_m$ with $m\neq i,j,k $. In this way we 
obtain $n + \binom{n}{3}$ expectation values from which we construct
our noncontextuality inequality,
\begin{widetext}
\beq \label{Ineqn}
\hspace{-0.5cm}\mathcal{I}_n & = & \alpha_n\left(\la A_1B_2B_3B_4\ldots B_n \ra + \la B_1A_2B_3B_4\ldots B_n \ra + \la B_1B_2A_3B_4\ldots B_n \ra + \ldots  + \la B_1B_2B_3B_4\ldots A_n \ra\right) \nonumber \\
&& - \la A_1A_2A_3B_4\ldots B_n \ra - \la A_1A_2B_3A_4\ldots B_n \ra - \ldots  -
\la B_1\ldots B_{n-3}A_{n-2}A_{n-1}A_{n} \ra \le \eta_C^n < \eta_Q^n = \alpha_n n+\binom{n}{3}, 
\eeq
\end{widetext}
where the constant $\alpha_n=\binom{n-1}{2}$ has been added for further convenience. 

It is not difficult to see that for the case $n=3$ the above inequality reproduces the one in Eq. \eqref{Ineq3}. While, as already mentioned, for $n=3$ it is equivalent to the MABK inequality \cite{PhysRevA.46.5375,Uspekhi,Mer90} if the commutation relations between the observables are satisfied due to the spatial separation between the subsystems, for $n>3$ this is not the case. The number of terms in the MABK Bell-type inequalities grows exponentially with $n$, whereas in our noncontextuality inequalities this number scales only polynominally with $n$. In Refs. \cite{TGB06,GC08,BAS+20} other Bell-type inequalities were designed for the graph states which again scale exponentially or linearly, thus, they differ from our inequalities. Our inequalities (\ref{Ineqn}) are designed such that they are suitable for the purpose of certification.
It is also important to notice that our inequality is constructed in such a way that for every three different correlators that enter $\mathcal{I}_n$ with $+$ and the associated 'negative' one, the non-common observables appearing in all these four correlators adopt the compatibility structure from the simplest inequality for $n=3$. To illustrate this with an example let us consider the inequality for $n=4$, 
\begin{widetext}
\beq \label{Ineqn4}
\mathcal{I}_4 & = & 3\left(\la A_1B_2B_3B_4 \ra + \la B_1A_2B_3B_4 \ra + \la B_1B_2A_3B_4 \ra + \la B_1B_2B_3A_4 \ra \right)\nonumber \\
&& - \la B_1A_2A_3A_4 \ra - \la A_1B_2A_3A_4 \ra - \la A_1A_2B_3A_4 \ra - \la A_1A_2A_3B_4 \ra \le \eta^4_C < \eta^4_Q = 16.
\eeq
\end{widetext}
Fig. \ref{figures2and3} presents the compatibility structures of the common observables for two choices of such four-element subsets of expectation values in $\mathcal{I}_4$. The first subset consists of the first three terms with $+$ sign and the last one with $-$ sign, all containing observables $A_1$, $A_2$ and $A_3$, whereas the second set is composed of the first, the second and the fourth '$+$' terms in $\mathcal{I}_4$ and the third 'negative one', all of them containing $A_1$, $A_2$ and $A_4$. 

\begin{figure}
    \centering
    \includegraphics[scale=0.27]{comp_graph.pdf}
    \includegraphics[scale=0.27]{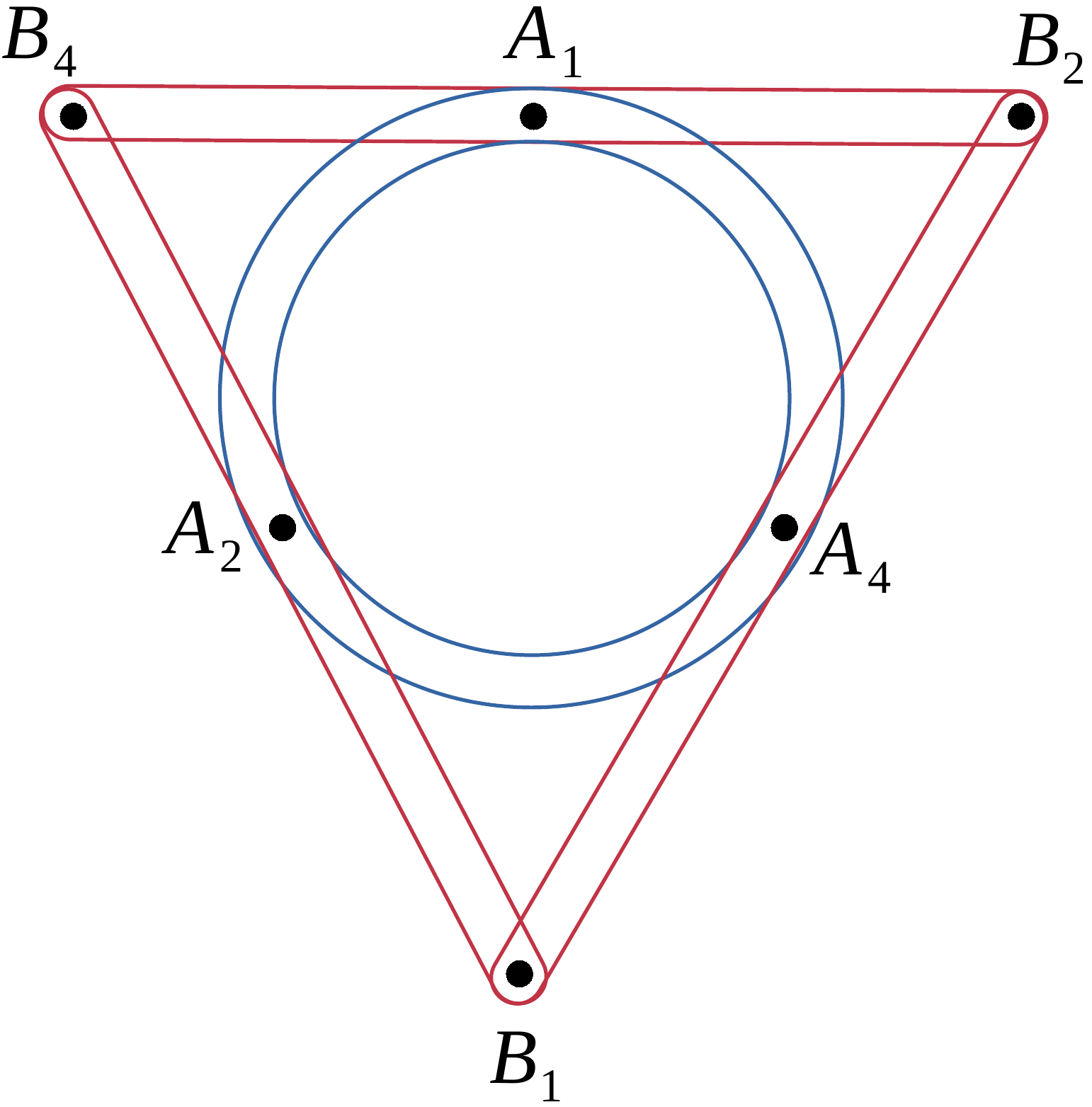}
    \caption{Hypergraphs of compatibility of subsets of observables related with two choices of subsets of correlators in Eq. (\ref{Ineqn4}). On the left we have the hypergraph associated with the subscripts 1, 2 and 3, and on the right side with the subscripts 1, 2 and 4.  These two hypergraphs also represent the compatibility structures of the observables corresponding to the simpler inequality (\ref{Ineq3}) and another such simpler inequality with the observables $A_i$ and $B_j$, with $i,j=1,2,4$, respectively. These two compatibility structures serve as the building blocks to construct the inequality (\ref{Ineqn4}).} \label{figures2and3}
    %\textcolor{red}{In this pictorial representation, every hyperedge represents a subset of observables that belong to a maximal context of the scenario. A maximal context is the context with the maximum subset of observables of the whole scenario, note that any subset of observables of a maximal context also belongs to a context. This figure doesn't contain the structure of the maximal contexts of the scenario, the purpose here is to exhibit the building blocks behind the construction of the inequality.} }
   \end{figure}

The inequality \eqref{Ineqn} is non-trivial for any $n$, i.e., $\eta^n_C < \eta^n_Q$. 
To prove this statement, let us first notice that its quantum bound is $\eta^n_Q = n\alpha_n+\binom{n}{3}$ and can be attained 
by the following observables
\beq\label{obsn}
&A_i = X_i, \qquad
B_j = Z_j, &
\eeq
and the unique graph state $|G_n\ra$ associated to the complete graph of $n$ qubits and stabilized by the operators in Eq. (\ref{ngraph}).
In fact, plugging (\ref{obsn}) into the expression $\mathcal{I}_n$
one realizes that all correlators with $+$ correspond to 
the stabilizing operators $G_i$, whereas those entering 
$\mathcal{I}_n$ with minus sign correspond to products of triples of
different $G_i$'s.

Let us then estimate the maximal classical value $\eta_C^n$, 
and, for pedagogical purposes, we first consider
the case $n=4$ for which the expression $\mathcal{I}_4$ can be
stated as
\begin{widetext}
\begin{eqnarray}
    \mathcal{I}_4&=&(\la A_1B_2B_3B_4 \ra + \la B_1A_2B_3B_4 \ra + \la B_1B_2A_3B_4 \ra-\la A_1A_2A_3B_4 \ra)\nonumber\\
    &&+(\la A_1B_2B_3B_4 \ra + \la B_1A_2B_3B_4 \ra + \la B_1B_2B_3A_4 \ra - \la A_1A_2B_3A_4 \ra)\nonumber\\
    &&+(\la A_1B_2B_3B_4 \ra + \la B_1B_2A_3B_4 \ra + \la B_1B_2B_3A_4 \ra  - \la A_1B_2A_3A_4 \ra)\nonumber\\
    &&+( \la B_1A_2B_3B_4 \ra +\la B_1B_2A_3B_4 \ra + \la B_1B_2B_3A_4 \ra - \la B_1A_2A_3A_4 \ra),
\end{eqnarray}
\end{widetext}
where each line in the right hand side of the above equation corresponds to a lifting of $I_3$ to $n=4$. 
Due to the fact that in each line we have basically the inequality for 
$n=3$, it is not difficult to see that for noncontextual models, $|\mathcal{I}_4|\leq 4\times 2=8$, which is clearly smaller than the maximal quantum value
$\eta_Q^4=16$. To prove that the same holds true for any $n$, 
it suffices to notice that analogously to $\mathcal{I}_4$, 
$\mathcal{I}_n$ can be rewritten as a sum of $\binom{n}{3}$ 
terms that are liftings of $\mathcal{I}_3$, and thus $\eta_C^n\leq 2\binom{n}{3}$.
At the same time, $\eta_Q^n=n\binom{n-1}{2}+\binom{n}{3}$ and thus 
$\eta_Q^n>\eta_C^n$ for any $n\geq 3$.

%As we proven for $n=3$, the non-triviality follows from the fact that if assume the %positive correlators to be $+1$, the negative correlators will be $+1$ also and this %lead to a contradiction when we try to find a classical strategy that attains the %algebraic value. 

\subsection{Certification based on the noncontextuality inequality}

Let us now show how the above inequality can be used for certification of 
the complete graph state and $n$ pairs of anticommuting observables. 
To this aim, we assume that a state $\ket{\psi}\in\mathcal{H}$
together with a set of $2n$ dichotomic observables $A_i$ and $B_i$ acting 
on $\mathcal{H}$ maximally violate (\ref{Ineqn}).
Then, as in the case $n=3$, this implies the following set of $n+\binom{n}{3}$ equations:
\begin{equation}\label{plusi}
    B_1\ldots B_{i-1}A_iB_{i+1}\ket{\psi}=\ket{\psi}
\end{equation}
with $i=1,\ldots,n$, and
\begin{equation}\label{minusi}
    B_1\ldots A_i\ldots A_j\ldots A_k\ldots B_n\ket{\psi}=-\ket{\psi}
\end{equation}
for any choice of $i,j,k=1,\ldots,n$ such that $i\neq j\neq k$.

%\beq \label{qrealization1na}
%\textcolor{blue}{A_1B_2B_3}B_4\ldots B_n |\psi\ra &=& |\psi\ra \\ %\label{qrealization2na}
%\textcolor{blue}{B_1A_2B_3}B_4\ldots B_n |\psi\ra &=& |\psi\ra \\ %\label{qrealization3na}
%\textcolor{blue}{B_1B_2A_3}B_4\ldots B_n |\psi\ra &=& |\psi\ra \\ %\label{qrealization4na}
%\textcolor{blue}{A_1A_2A_3}B_4\ldots B_n |\psi\ra &=& -|\psi\ra 
%\eeq
%and
%\beq
%\textcolor{blue}{A_1B_2}B_3\textcolor{blue}{B_4}\ldots B_n |\psi\ra &=& |\psi\ra \\
%\textcolor{blue}{B_1A_2}B_3\textcolor{blue}{B_4}\ldots B_n |\psi\ra &=& |\psi\ra \\
%\textcolor{blue}{B_1B_2}B_3\textcolor{blue}{A_4}\ldots B_n |\psi\ra &=& |\psi\ra \\ 
%\textcolor{blue}{A_1A_2}B_3\textcolor{blue}{A_4}\ldots B_n |\psi\ra &=& -|\psi\ra
%\eeq
%\sBox{\cj{What is the reason to divide the eigenvalue equations into the two set of equations as above? Rafa: In order to help the reader to understand the construction of the inequality. }}
%
As a consequence of these equations, we have
\begin{equation} \label{anticomut1}
A_1B_1 |\psi\ra = A_1 A_2 B_3 B_4\ldots B_n |\psi\ra = A_2 B_2 |\psi\ra, 
\end{equation}
%
%= -B_3 A_3 |\psi\ra \nonumber \\ 
%&= A_1 A_3 B_2 B_4\ldots B_n |\psi\ra = -B_2 A_2 |\psi\ra = A_3 B_3 |\psi\ra,
%\end{align}
%
where the first and the second equality stem from Eq. (\ref{plusi}) for $i=2$ and
Eq. (\ref{plusi}) for $i=1$, respectively. Then, by applying 
Eq. (\ref{minusi}) for $i=1$, $j=2$ and $k=3$, one obtains
\begin{equation} \label{anticomut2}
A_1B_1 |\psi\ra =  -B_3 A_3 |\psi\ra. 
\end{equation}
On the other hand, using Eq. (\ref{plusi}) with $i=3$ 
we can write
\begin{equation}\label{Wawel1}
    A_1B_1 |\psi\ra=A_1 A_3 B_2 B_4\ldots B_n |\psi\ra=A_3 B_3 |\psi\ra,
\end{equation}
where the second equation is a consequence of Eq. (\ref{plusi}) for $i=1$.
Simultaneously, an application of Eq. (\ref{minusi})
for $i=1$, $j=2$ and $k=3$ to the second terms in the above gives
\begin{equation}\label{Wawel2}
    A_1B_1 |\psi\ra=-B_2 A_2 |\psi\ra.
\end{equation}

Note that our inequality is designed in a such way that it is symmetric under any permutation of subscripts, i.e., it is invariant under the transformations  $A_i \leftrightarrow A_j$ together with $B_i \leftrightarrow B_j$. This when applied to  Eqs. (\ref{anticomut1})-(\ref{Wawel2}) results in the following relations
\beq \label{anticomutgeneral}
A_i B_i |\psi\ra = A_j B_j |\psi\ra = - B_i A_i |\psi\ra
\eeq
for $i,j \in \{1,2,\ldots ,n\}$. In particular, $A_i$ and $B_i$
anticommute, 
\begin{equation}\label{anticommutn}
    \{A_i,B_i\}\ket{\psi}=0 \qquad (i=1,\ldots,n).
\end{equation}

Having established the key relations between the state 
and the observables, let us now, analogously to 
the case $n=3$, identify a subspace of the Hilbert space $\mathcal{H}$ 
which is invariant under the action of all observables $A_i$
and $B_i$. The subspace is given by
\begin{eqnarray}\label{subInvodd2}
    V_{n}&=&\mathrm{span}\{\ket{\psi},B_{i_1}\ket{\psi},B_{i_1}B_{i_2}\ket{\psi},\ldots, B_{i_1}B_{i_2}\ldots B_{i_{k}}\ket{\psi},\nonumber\\
    &&\hspace{0.3cm}\ldots,B_{i_1}B_{i_2}\ldots B_{i_{n-1}}\ket{\psi},B_{1}\ldots B_{n}\ket{\psi}\},
\end{eqnarray}
where $i_j=1,\ldots,n$ for any $j$ and $i_1<i_2<\ldots<i_k<\ldots <i_{n-1}$.

For instance, in the simplest cases of $n=3$ and $n=4$, 
the above construction gives 
\begin{equation}
    V_3=\mathrm{span}\{\ket{\psi},B_i\ket{\psi},B_{i}B_{j}\ket{\psi},B_1B_2B_3\ket{\psi}\}
\end{equation}
and
\begin{equation}
    V_4\!=\!\mathrm{span}\{\ket{\psi},B_i\ket{\psi},B_iB_j\ket{\psi},B_iB_jB_k\ket{\psi},B_1B_2B_3B_4\ket{\psi}\}
\end{equation}
with $i,j,k=1,\ldots,n$ and $i<j<k$. In particular, 
$V_3$ is exactly the same as the one defined in Eq. 
(\ref{defInv}); due to Eq. (\ref{plusi}),
$B_iB_j\ket{\psi}=A_k\ket{\psi}$ with $i\neq j\neq k$
and $B_1B_2B_3\ket{\psi}=A_1B_1\ket{\psi}$.

Let us then notice that the number of vectors 
spanning $V_n$ is $2^n$. This is because each subset 
of vectors of the form $A_{i_1}\ldots A_{i_k}\ket{\psi}$ for 
$i_1<\ldots <i_k$ contains $\binom{n}{k}$ elements, 
and we have $n+1$ such subsets indexed by $k=0,\ldots,n$. 
Thus the total number of vectors can be
counted as 
\begin{equation}
\sum_{k=0}^{n}\binom{n}{k}=(1+1)^n=2^n, 
\end{equation}
meaning that 
%
%\begin{equation}
 $   \dim V_n\leq 2^n$.
%\end{equation}
%
In fact, as we show later  $\dim V_n$ is exactly $2^n$.

%In the particular case of $n=3$, the invariant subspace was constructed with a %subset of elements of these $2^{2n}$ elements. Here, $V_n$ is spanned in the %definition by a bigger set of elements and the maximum we can say at this %point is about an upper bound to the dimension of $V_n$. 
Our aim now is to identify the form of the operators $A_i$ and $B_j$ projected onto the subspace $V_n$. The idea of the proof of self-testing is similar to those we used in the case $n=3$.

\begin{lemma}\label{invsubVn}
The subspace $V_n$ of $\mathcal{H}_n$ is invariant under the action of the operators $A_i$ and $B_j$ for $i,j=\{ 1,2,\ldots,n \}$.
\end{lemma}
\begin{proof}
It can be checked that the action of any operator $A_i$ or $B_j$ over all the set of $2^{n}$ elements that generate the subspace $V_n$ is, up to the factor $-1$, a permutation over this set. Indeed, 
the action of $B_m$ on vectors of the form 
$B_{i_1}\ldots B_{i_k}\ket{\psi}$ with $i_1<\ldots<i_k$ for $k=1,\ldots,n$ returns vectors of a similar form with $k\to k-1$ if $m$ equals one of the subscripts
$i_1,\ldots,i_k$, or with $k\to k+1$ otherwise. In both cases
the resulting vectors are already in $V_n$. 

Let us then consider the $A_i$ observables. Due to Eq. (\ref{plusi}) and taking
into account the commutation (\ref{commutn}) or anticommutation (\ref{anticommutn}) their 
action on the vectors spanning $V_n$ can always be represented as an action of 
a product of $n-1$ different $B_i$ observables. Thus, 
when applied to $B_{i_1}\ldots B_{i_k}\ket{\psi}$ with $i_1<\ldots<i_k$ 
for $k=1,\ldots,n$ they will again produce
vectors involving products of $B_i$ operators that are already in $V_n$. This completes the proof.
\end{proof}

%
\iffalse
\begin{lemma}\label{invsubVn}
The subspace $V_n$ of $\mathcal{H}_n$ is invariant under the action of the operators $A_i$ and $B_j$ for $i,j=\{ 1,2,\ldots,n \}$.
\end{lemma}
%
\begin{proof}
It can be easily checked that the action of any operator $A_i$ or $B_j$ over all the set of $2^{2n}$ elements that generate the subspace $V_n$ is another element of this set, up to a phase factor of $\pm 1$. Since we assume commutativity among convenient measurements, it is trivial to check the action of all operators, except  for $B_j$ when $k_j=1$ in the elements spanning $V_n$ in Eq.\eqref{invsubsVn1}. However, in this exceptional case the action of $B_j$ over such elements will be another element that span $V_n$ with the negative signal which can be seen as a direct consequence of the Eq.\eqref{anticomutgeneral}. Therefore, $V_n$ is an invariant subspace.
\end{proof}
\fi
%\sBox{\cj{The above proof may be written elaborately as it seems unclear to understand.}}

This is a key step of our considerations because taking into account the fact that $A_i$ and $B_i$ are quantum observables, Lemma \ref{invsubVn} implies that they can be represented as a direct sum of two blocks, 
\begin{equation}
A_i = \hat{A}_i \oplus A_{i}',\qquad  B_j = \hat{B}_j \oplus B_{j}',
\end{equation}
where $\hat{A}_i$ and $\hat{B}_i$ are projections of $A_i$ and $B_i$ onto $V_n$, 
that is $\hat{A}_i=P_n A_i P_n$ and $\hat{B}_i=P_n A_i P_n$ with $P_n:\mathcal{H}_n\to V_n$
denoting the projector onto $V_n$. On the other hand, $A_{i}'$ and $B_{j}'$
are defined on the orthocomplement of $V_n$ in the Hilbert space $\mathcal{H}_n$ that we denote
$V_n^{\perp}$; clearly, $\mathcal{H}_n=V_n\oplus V_n^{\perp}$. 

Importantly, $A_i'$ and $B_i'$ act trivially on the subspace $V_n$, in particular $A_i'\ket{\psi}=B_i'\ket{\psi}=0$, and consequently
it is enough for our purposes  to characterize $\hat{A}_i$ and $\hat{B}_j$. Our first step to achieve this goal is to prove the following lemma.
\begin{lemma}\label{lemma2n}
$\{ \hat{A}_i,\hat{B}_i \} = 0$ for all $i\in\{ 1,\ldots,n \}$.
\end{lemma}
\begin{proof}
In order to prove this statement we will show that $\{A_i,B_i\}V_n=0$, that is, 
all these anticommutators act trivially on any vector from $V_n$.
To this aim, let us investigate how $\{A_i,B_i\}$ acts 
on the vectors spanning $V_n$. We first see that $ \{ A_i,B_i \} |\psi\ra = 0$ as a direct consequence of Eq. \eqref{anticomutgeneral}. Let us then consider vectors $B_j\ket{\psi}$. If $i\neq j$, we can directly use the commutation relations (\ref{commutn}) to write
\begin{equation}
    \{A_i,B_i\}B_j\ket{\psi}=B_j\{A_i,B_i\}\ket{\psi}=0,
\end{equation}
On the other hand, if $i=j$, one has
\begin{equation}
    \{A_i,B_i\}B_i\ket{\psi}=(A_i+B_iA_iB_i)\ket{\psi}=0,
\end{equation}
where the last equality is again a consequence of the facts that $A_iB_i\ket{\psi}=-B_iA_i\ket{\psi}$
and that $B_i^2=\mathbbm{1}$. 

It is not difficult to realize that the above reasoning extends to 
any vector spanning the subspace $V_n$. Indeed, 
let us consider vectors of the form $B_{i_1}\ldots B_{i_k}\ket{\psi}$ with $i_1<\ldots<i_k$ for $k=1,\ldots,n$ and assume first that all $i_1,\ldots,i_k$
differ from $i$. Then, due to the commutation relations, one directly has
\begin{equation}
    \{A_i,B_i\}B_{i_1}\ldots B_{i_k}\ket{\psi}=B_{i_1}\ldots B_{i_k}\{A_i,B_i\}\ket{\psi}=0.
\end{equation}
On the other hand, if one of the subscripts, say $i_1$, equals $i$, then 
\begin{eqnarray}
    \{A_i,B_i\}B_{i_1}\ldots B_{i_k}\ket{\psi}&=&B_{i_2}\ldots B_{i_k}\{A_i,B_i\}B_1\ket{\psi}\nonumber\\
    &=&B_{i_2}\ldots B_{i_k}(A_i+B_iA_iB_i)\ket{\psi}\nonumber\\
    &=&0,
\end{eqnarray}
where the last equality follows from the fact that $A_i$
and $B_i$ anticommute when acting on $\ket{\psi}$ and from $B_i^2=\mathbbm{1}$. 

Taking into account that each $\{A_i,B_i\}$ is a Hermitian operator and 
that it acts trivially on the whole subspace $V_n$, 
one directly concludes that $\{\hat{A}_i,\hat{B}_i\}=0$. 
\end{proof}

%\sBox{\cj{I think that the above proof should be rewritten.}}

Our next lemma is a straightforward generalization of Lemma \ref{lemma3}.
\begin{lemma}\label{lemma3n}
Suppose the maximal quantum violation of our inequality (\ref{Ineqn}) is observed.
Then, there exists a basis of $V_n$ for which
\beq \label{Aitilden}
\hat{A}_i &=& X_i, \qquad \hat{B}_j = Z_j.
\eeq
\end{lemma}
\begin{proof}
We will proceed recursively starting from the pair $\hat{A}_1$ and $\hat{B}_1$.
It follows from Lemma \ref{lemma2n} that $\{ \hat{A}_1,\hat{B}_1 \} = 0$ which means that the dimension $d=\dim V_n$ an even number, i.e., $d=2k$ for some $k=1,\ldots,2^{n-1}$
(recall that $d\leq 2^n$) and that there exists a unitary $U_1$ acting on $V_n$ such that
\beq
U_1^{\dagger}\, \hat{A}_1\, U_1 &=& X \otimes \mathbbm{1}_k, \\
U_1^{\dagger}\, \hat{B}_1\, U_1&=& Z \otimes \mathbbm{1}_k,
\eeq
where $\mathbbm{1}_k$ is an identity acting on $\mathbbm{C}^{k}$.

Next, to determine the remaining observables $\hat{A}_i$ and $\hat{B}_i$
we exploit the fact that they all must commute with both $\hat{A}_1$ and $\hat{B}_1$.
With this aim, we use the fact that $V_n=\mathbbm{C}^2\otimes\mathbbm{C}^k$ to decompose $\hat{A}_i$ and $\hat{B}_i$ $(i=2,\ldots,n)$ in terms of the Pauli basis as
\begin{equation}
   U_1^{\dagger}\, \hat{R}_i\, U_1=\mathbbm{1}\otimes M_0^R+X\otimes M_1^R+Y\otimes M_2^R+Z\otimes M_3^R,
\end{equation}
where $R=A,B$, and $M_i^R$ are some Hermitian matrices acting on $\mathbbm{C}^k$.
Now, $[\hat{R}_i,\hat{A}_1]=0$ implies that $M_2^R=M_3^R=0$, whereas
from $[\hat{R}_i,\hat{B}_1]=0$ one concludes that $M_1^R=0$. As a result,
all $\hat{A}_i$ and $\hat{B}_i$ with $i=2,\ldots,n$ admit
the following representation,
\begin{equation}
  U_1^{\dagger}\,  \hat{A}_i\,U_1=\mathbbm{1}_2\otimes M_i,\qquad U_1^{\dagger}\, \hat{B}_i\, U_1=\mathbbm{1}_2\otimes N_i,
\end{equation}
where $M_i$ and $N_i$ act on $\mathbbm{C}^k$; in fact, they are 
Hermitian and obey $M_{A}^2=M_B^2=\mathbbm{1}_k$, and thus are quantum observables.
Moreover, it follows from Lemma \ref{lemma2n} that $\{M_i,N_i\}=0$ for all $i=2,\ldots,n$. 

We have thus a set of $2(n-1)$ quantum observables $M_i$ and $N_i$ that satisfy the same commutation and anticommutation relations as $\hat{A}_i$ and $\hat{B}_i$, 
and therefore we can use the above reasoning again to
conclude that the dimension $k$ is even, that is, $k=2k'$ for some $k'=1,\ldots,2^{n-4}$, and that there is a unitary operation $U_2:\mathbbm{C}^k\to \mathbbm{C}^k$ such that 
\begin{eqnarray}
    U_2^{\dagger}\,M_2\, U_2=X\otimes \mathbbm{1}_{k'},\qquad 
    U_2^{\dagger}\,N_2\, U_2=Z\otimes \mathbbm{1}_{k'},
\end{eqnarray}
where $\mathbbm{1}_{k'}$ is an identity acting on $\mathbbm{C}^{k'}$.
We then use the fact that the other operators $M_i$ and $N_i$ with $i=3,\ldots,n$ commute with both $M_2$ and $N_2$ to see that they are of the form
$M_i=\mathbbm{1}_2\otimes P_i$ and $N_i=\mathbbm{1}_2\otimes Q_i$, where
$P_i$ and $Q_i$ are quantum observables acting on $\mathbbm{C}^{k'}$.

It is now clear that the above procedure can be applied iteratively many times
until all the observables are proven to be of the form (\ref{Aitilden}), 
of course, up to certain unitary operation. In fact, one finds that 
$V_n=(\mathbbm{C}^2)^{\otimes n}$, that is, it is an $n$-qubit Hilbert space. 
Moreover, there is a unitary operation $U$ composed of all the 
intermediate unitary operations $U_i$ such that 
\begin{equation}
    U^{\dagger}\hat{A}_i\, U = X_i,\qquad U^{\dagger}\hat{B}_i\, U = Z_i
\end{equation}
for any $i=1,\ldots,n$.
\end{proof}

One of the main messages that one takes from this lemma is that 
the dimension of $V_n$ is exactly $2^n$; in other words, $V_n$ is 
isomorphic to an $n$-qubit Hilbert space. In this sense our 
inequalities can be seen as dimension witnesses: 
the dimension of the Hilbert space supporting a state and observables giving rise to the maximal violation of our inequalities must be at least $2^n$.
Moreover, the above lemma implies that 
a set of $n$ pairs of anti-commuting quantum observables with 
outcomes $\pm1$ that satisfy the commutation relations (\ref{commutn}) 
can always be represented, up to a single unitary operation, as a
tensor product of single-qubit operators (\ref{Aitilden}).
We have thus arrived at our main result.

\begin{thm}\label{Theonqubit}
If a quantum state $|\psi \ra$ and a set of dichotomic observables $A_i$ and $B_j$ with $i,j=\{ 1,2,\ldots,n \}$ give rise to maximal violation of the inequality \eqref{Ineqn}, then there exist a projection $P_n:\mathcal{H}_n \rightarrow \mathbbm{C}^{2^n}$ and a unitary $U$ acting on $\mathbbm{C}^{2^n}$ such that
\beq
U^\dagger (P\, A_i\, P^\dagger) U &=& \hat{A}_i, \\
U^\dagger (P\, B_j\, P^\dagger) U &=& \hat{B}_j, \\
U (P|\psi \ra) &=& |G_n \ra,
\eeq
where $\hat{A}_i$ and $\hat{B}_j$ are defined in Eq. \eqref{Aitilden} and $|G_n \ra$ is the complete graph state of $n$ qubits. 
\end{thm}
\begin{proof}
The state $|\psi \ra\in\mathcal{H}_n$ and observables $A_i,B_j$ acting on the Hilbert space $\mathcal{H}_n$ attain the maximal quantum violation of the inequality \eqref{Ineqn} if, and only if, they satisfy the set of $n+\binom{n}{3}$ equations  \eqref{plusi} and \eqref{minusi}. The algebra induced by this set of equations allows us to prove Lemmas \ref{invsubVn}, \ref{lemma2n} and \ref{lemma3n}; in particular, it follows that there exists a projection $P_n:\mathcal{H} \rightarrow V_n \cong \mathbbm{C}^{2^n} $ and a unitary $U$ acting on $V_n$ such that
\beq
U^\dagger (P\, A_i\, P^\dagger) U &=& X_i, \\
U^\dagger (P\, B_j\, P^\dagger) U &=& Z_j,
\eeq
In this way, the products of observables that appear in the first $n$ correlators  in the inequality \eqref{Ineqn} give stabilizing operators of the graph state associated to the complete graph with $n$ vertices, whereas the correlators with negative sign correspond to products of three different stabilizing operators for which the graph state $\ket{G_n}$ is an a eigenvector with associated eigenvalue $-1$. Thus the complete graph state will be the unique state that attain the maximal quantum violation, then
\beq
U (P|\psi \ra) = |G_n \ra.
\eeq
This completes the proof.
\end{proof}

\iffalse
\subsection{A semi device-independent approach}

An important comment that has to be done here is about the relation of the inequality \eqref{Ineqn} designed in this section with the semi-device independent approach we proposed in the previous section. With reasonable assumptions about the physical subsystems where the measurements are conveniently performed, a similar approach can be naturally adapted in order to certify any multi qubit graph state.

Now we elucidate how to generalize this approach to certify any graph state. Just see that in order to find the convenient measurements for a given graph state we just need to construct a convenient table as we did in the Subsection\eqref{sdia} for the graph state associated to the graph which is not isomorphic to the complete graph. There is a simple rule to construct such table for the general case. In the case of a complete graph, the $A_i$ operators have only $X$ operators and for every edge withdraw, we have to add $Z$ operators conveniently. Starting from the complete graph, in this way we can construct a convenient set of measurements for any graph state. We claim that the self-testing holds with the inequality we designed.
\fi

\section{Robustness}\label{Robustness}
Here we obtain  fidelity bounds on the state and measurements leading to the given  nonmaximal violation of 
the inequality \eqref{Ineqn} to demonstrate that our scheme is robust  to
errors and experimental imperfections. 
For simplicity, let us focus on the case of $n=3$. 
Let us say that the maximal quantum violation of the inequality (\ref{Ineq3}) is observed with an $\epsilon$-error, i.e., 
a nonmaximal violation of $4-\epsilon$ is observed. Then, the correlators satisfy 
the following bounds:
\begin{align}\label{eq: Error statistics n=3}
\begin{split}
    \bra{\psi}A_1B_2B_3 \ket{\psi}&\ge 1-\epsilon,\\
    \bra{\psi}B_1A_2B_3 \ket{\psi}&\ge 1-\epsilon,\\
    \bra{\psi}B_1B_2A_3 \ket{\psi}&\ge 1-\epsilon,\\
    -\bra{\psi}A_1A_2A_3 \ket{\psi}&\ge 1-\epsilon,
\end{split}    
\end{align}
for some  $\epsilon > 0$. We demonstrate that for a small enough value of $\epsilon$, the quantum realization is close enough to the optimal quantum realization which leads to the maximal quantum violation of the inequality. This is the purpose of robustness analysis that will be presented here, i.e., we show that in the limit $\epsilon \rightarrow 0$ the quantum realization is close to the optimal one.

In the presence of errors, it is not straightforward to guarantee the existence of
an invariant subspace under the action of the operators $A_i$ and $B_j$, as we have in the self-testing of the optimal quantum realization. However, the robustness of the protocol can be demonstrated, in a similar way to that of Ref. \cite{IMO+20}, by proving the existence of an eight-dimensional ideal subspace $\hat{V}$ together with a state $\ket{\hat{\psi}}\in\hat{V}$ and observables $\hat{A}_i$ and $\hat{B}_i$ acting on it
such that their fidelities with the actual state and measurements approach one in the limit of $\epsilon\to 0$.
We define the state fidelity as $F(|\hat{\psi} \ra,|\psi\ra):=|\langle\hat{\psi} | \psi \ra |^2$, and the operator fidelity as $F(\hat X_i,X_i):=(1/8)\mathrm{Tr}(\hat X_i X_i)$  (with $X_i=A_i$),  where the $1/8$ factor is used to normalize the fidelity since $\mathrm{Tr}(\hat X_i X_i) \le 8 $,
and similarly defined between $\hat B_j$ and $B_j$. 
%
%The fidelities in both the state and measurements induce a measure in the abstract space %of state and measurements, such measure induces a notion of continuity that appear %naturally in the definition of robustness. 
%
%For our convenience we chose these expressions of fidelity, anyway any other measure will %induce the same notion of continuity since the space in study is finite dimensional.
%
%This comparison comes with a metric induced between these sets, where we used the %fidelity for the state and measurements. 
%
%The fidelity goes to $1$ in the limit when $\epsilon$ goes to $0$, i.e., the robustness %arises. 
Formally, we have the following theorem.
\begin{thm}\label{thm robust}
Suppose a quantum state $|\psi \ra$ and a set of measurements $A_i,B_j$ with $i,j=\{ 1,2,3 \}$ in a Hilbert space $\cH$ satisfy the ideal expectations corresponding to the maximal quantum violation of the inequality \eqref{Ineq3} to within error     $\epsilon$.  Then there exists a projection $P:\cH \rightarrow \hat V$, where $\dim(\hat{V})=8$, a state $|\hat{\psi} \ra \in \hat{V}$ and $\hat A_i$, $\hat{B}_j$ which are Hermitian involutions acting on $\hat{V}$ for all $i$ and $j$ such that 
\beq \label{qrealization1n}
\la \hat \psi |  \hat A_1 \hat B_2 \hat B_3 |\hat\psi\ra &=&  1 \nonumber \\ 
\la \hat \psi | \hat B_1 \hat A_2 \hat B_3 |\hat\psi\ra &=&  1 \nonumber \\ 
\la \hat  \psi | \hat B_1 \hat B_2 \hat A_3  |\hat\psi\ra &=&  1 \nonumber
\eeq
and there exists some unitary $U$ acting on $ \hat V $ such that 
\begin{align*}
     F(U | \hat \psi\rangle  ,| \psi \ra )&\ge 1- \epsilon_0,\\
    F( U \hat A_{i}  U^{\dagger}, A_{i})&\ge 1- \epsilon_1 \quad \forall i, \\
    F( U  \hat B_{j}  U^{\dagger}, B_{j})&\ge 1- \epsilon_2 \quad \forall i , 
\end{align*}
where  $\epsilon_0=25 \epsilon$, $\epsilon_1=0$, $\epsilon_2=4 \epsilon$.
\end{thm}
%The above robustness also holds in the case of the temporal inequality provided that all ideal expectation values associated with this inequality are satisfied  to within error  $\epsilon$.
The proof of the above theorem is given in Appendix \ref{robustproof}. This theorem implies that there exist a subspace in which, up to a small enough error, the quantum realization leading to the nonmaximal quantum violation is close to the optimal quantum realization up to a unitary. For instance, an error of $0.1 \%$ in each expectation value implies that the state fidelity is not less than $97.5 \%$ and the operator fidelities of $B_j$ are not less than $99.6 \%$.
Following the proof of Theorem \ref{thm robust}, one can also obtain the fidelity bounds for any $n$ demonstrating the robustness
of the scheme similarly as presented in Theorem \ref{thm robustn}.

In order to obtain a tight self-testing bound that is applicable to more noisy practical scenario in which the given nonmaximal violation is not almost perfect, one may employ numerical method to bound the state fidelity as a function of violation of the noncontextual inequality based on semidefinite programming relaxations of quantum contextual sets introduced in Ref. \cite{CFW21} or  analytical method based on operator inequalities introduced in Ref. \cite{Kan16}.

\section{Conclusions}

Quantum contextuality provides a notion of nonclassicality for single systems. Motivated by extending the task of self-testing based on Bell inequalities to scenarios where entanglement is not necessary or spatial separation between the subsystems is not required, self-testing of quantum devices based on quantum contextuality has recently been explored. In this work we have followed this research direction and have introduced a family of inequalities revealing quantum contextuality and have shown that they can be used for certification of multiqubit quantum systems.
An interesting feature of our scheme is that it is scalable: the amount of information about the observed nonclassical correlations needed to certify the underlying quantum system grows only polynomially with the number of qubits that are certified.

Such contextuality-based certification schemes rely, however, on compatibility 
relations between the involved measurements and are thus generally difficult to implement in practice. A natural follow-up of this work would be therefore to 
design a scheme for certification built on our inequalities that does not rely on 
the compatibility relations, but rather allows one to deduce 
them from the observed nonclassicality. 
Such schemes for single quantum systems have recently been 
proposed in Refs. \cite{SSA20,MMJ+21} within the sequential measurements
or temporal correlations scenario. It would be thus interesting 
to see whether our results can be mapped to this scenario. 
Another possible direction for further research is to 
explore whether one can improve the scalability of our scheme 
with the number of the certified qubits. From a general perspective 
it is a highly nontrivial question what is the minimal information
about the observed nonclassical correlations that enables making
nontrivial statements about the underlying quantum system.

\section*{acknowledgments}

We thank Debashis Saha for discussions. We are also indebted to Owidiusz Makuta for providing us the example presented at the end of Section III.
This work was supported by the Foundation for Polish Science
through the First Team project (First TEAM/2017-
4/31) co-financed by the European Union under the European
Regional Development Fund.

\bibliography{bibliography}
\appendix

\section{Proof of Theorem \ref{thm robust}}\label{robustproof}
Here we present the derivation of robustness of the scheme  given in Theorem \ref{thm robust}. This will be similar to the proof of robustness of the certification scheme given in 
Ref. \cite{IMO+20}. 

For providing robust self-testing statements in the Bell tests where dichotomic measurements are implemented \cite{Kan16}, Jordan decomposition  for the state and measurements has been employed  to simplify the derivation. In Ref. \cite{IMO+20}, Jordan decomposition has also been extended to the contextuality scenario to provide robust certification of the two-qubit system. In Appendix~\ref{Appendix: Jordan}, we extend this Jordan decomposition to our scenario to provide the robust certification.  According to this decomposition, we can decompose the Hilbert space $\cH$
in which the state and measurements leading to the violation of our inequality acting  as 
\begin{equation} \label{Jdec}
\cH=\bigoplus_l\cH_l, 
\end{equation}
where each $\cH_l$ has dimension at most eight and invariant under the action of
$A_{i}$ and $B_{j}$. 

With respect to the decomposition (\ref{Jdec}), 
the state $\ket{\psi}$ can be written as 
\begin{equation}
    \ket{\psi}= \sum_l \sqrt{p_l}  \ket{\psi_l},
\end{equation} 
where  $\ket{\psi_l}  \in \cH_l$ 
and $\sum_l p_l =1$. We express each $\ket{\psi_l}$ 
in the computational basis $\{\ket{abc}_l: a,b,c \in \{0,1\} \}$  as
\begin{equation} \label{nonidealstate}
     \ket{\psi_l}=\sum_{a, b,c=0,1} c^{(l)}_{abc}\ket{a b c}_l, 
\end{equation}
where $c^{(l)}_{abc}$ satisfies
%$\sum_{a,b,c}|c^l_{abc}|^2=1$
$\sum_{a,b,c}|c^{(l)}_{abc}|^2=1$.

We define the subspace $\hat V \subseteq \cH $ as the linear span of $\{\ket{\tilde a \tilde b \tilde c} : = \sum_l \sqrt{p_l} \ket{abc}_l \}$.
We define the ideal state in this subspace  as
\begin{equation} \label{IdealStPr}
 \ket{\hat \psi } :=  \frac{1}{\sqrt{2}}\big(\ket{\tilde 0 \tilde 0 \tilde 0  }-\ket{\tilde 1 \tilde 1 \tilde 1}\big),
\end{equation}
which can be re-expressed as 
\begin{equation}
    \ket{\hat \psi } =  \sum_l \sqrt{p_l}  \ket{\hat \psi_l},
\end{equation}
where 
\begin{equation}\label{ProjIdealstate}
\ket{\hat \psi_l}:=\frac{1}{\sqrt{2}}(\ket{000}_l-\ket{111}_l).
\end{equation}
Note that for the ideal observables defined as 
\begin{align}\label{JordanObsn3}
\hat A_1&:=\bigoplus_l (\hat A_{1_l}\otimes\mathbbm{1}_l \otimes\mathbbm{1}_l ), \quad \hat B_1:=\bigoplus_l  (\hat B_{1_l}\otimes\mathbbm{1}_l \otimes\mathbbm{1}_l ) \nonumber \\
\hat A_{2}&:=\bigoplus_l (\mathbbm{1}_l\otimes \hat A_{2_l} \otimes\mathbbm{1}_l ), \quad  \hat B_{2}:=\bigoplus_l (\mathbbm{1}_l\otimes \hat B_{2_l} \otimes\mathbbm{1}_l ) \nonumber \\
\hat A_{3}&:=\bigoplus_l (\mathbbm{1}_l\otimes\mathbbm{1}_l\otimes \hat A_{3_l}  ), \quad \hat B_{3}:=\bigoplus_l (\mathbbm{1}_l \otimes\mathbbm{1}_l \otimes \hat B_{3_l}  ), \nonumber
\end{align}
where
%\begin{align}\label{eq: jordan angles}
%    \hat{A}_{1}&= XII,\ \hat{B}_{2}= (\cos \phi_l IYI+\sin \phi_l IXI),\notag\\
%    \hat{A}_{2}&= IXI,\ \hat{B}_{1}= (\cos \theta_l YII+\sin \theta_l XII),  \notag\\
%    \hat{A}_{3}&= IIX,\ \hat{B}_{3}= (\cos \nu_l IIY+\sin \nu_l IIX),
%\end{align}
%with $\theta_l,\phi_l, \nu_l \in[-\frac{\pi}{2},\frac{\pi}{2}]$ and $X$ and $Y$ are the Pauli operators.
\begin{align}
    \hat A_{i_l}&= X_i,\quad \hat  B_{j_l}= Y_j,
\end{align}
where $X_i$ and $Y_j$ are the Pauli operators acting on the $i$th qubit,
the ideal state defined in Eq. (\ref{IdealStPr}) violates the noncontextuality inequality (\ref{Ineq3}) maximally. Note that the ideal state projected onto $\cH_l$, i.e., 
$\ket{\hat \psi_l}$ has the form of the GHZ state (\ref{GHZstate}).  We have chosen the above particular form of the ideal state for our convenience.

We now express the nonideal observables with respect to the Jordan decomposition. 
%For our convenience, we choose the $Z$ basis on each $\cH_l$: $\{\ket{abc}_l: a,b,c \in \{0,1\} \}$ to express the state and measurements.
Due to the fact that we have three pairs of dichotomic observables that do not commute on the quantum state, which is a consequence of Lemma \ref{prop: Anticommutativity}, the dimension of each of the subspace $\cH_l $ in Eq. (\ref{Jdec}) can be taken to be eight. 
From Corollary \ref{corrNonideal}, it follows that 
\begin{align}
A_1&=\bigoplus_l (A_{1_l}\otimes\mathbbm{1}_l \otimes\mathbbm{1}_l ), \quad B_1=\bigoplus_l (B_{1_l}\otimes\mathbbm{1}_l \otimes\mathbbm{1}_l ) \nonumber \\
A_{2}&=\bigoplus_l (\mathbbm{1}_l\otimes A_{2_l} \otimes\mathbbm{1}_l ), \quad B_{2}=\bigoplus_l (\mathbbm{1}_l\otimes B_{2_l} \otimes\mathbbm{1}_l ) \nonumber \\
A_{3}&=\bigoplus_l (\mathbbm{1}_l\otimes\mathbbm{1}_l\otimes A_{3_l}  ), \quad B_{3}=\bigoplus_l (\mathbbm{1}_l \otimes\mathbbm{1}_l \otimes B_{3_l}  ), \nonumber
\end{align}
where by using 'local' unitary operations we can always choose ${A}_{i_l}$ and $B_{i_l}$ acting on  $\cH_l$  to be other following forms:
%\begin{align}\label{eq: jordan angles}
%    \hat{A}_{1}&= XII,\ \hat{B}_{2}= (\cos \phi_l IYI+\sin \phi_l IXI),\notag\\
%    \hat{A}_{2}&= IXI,\ \hat{B}_{1}= (\cos \theta_l YII+\sin \theta_l XII),  \notag\\
%    \hat{A}_{3}&= IIX,\ \hat{B}_{3}= (\cos \nu_l IIY+\sin \nu_l IIX),
%\end{align}
%with $\theta_l,\phi_l, \nu_l \in[-\frac{\pi}{2},\frac{\pi}{2}]$ and $X$ and $Y$ are the Pauli operators.
\begin{align}\label{eq: jordan angles}
    A_{1_l}&= X_1,\quad B_{1_l}= \cos \theta_l Y_1+\sin \theta_l X_1,\notag\\
    A_{2_l}&= X_2,\quad  B_{2_l}= \cos \phi_l Y_2+\sin \phi_l X_2,\notag\\
    A_{3_l}&= X_3,\quad B_{3_l}= \cos \nu_l Y_3+\sin \nu_l X_3,
\end{align}
with $\theta_l,\phi_l, \nu_l \in[-\frac{\pi}{2},\frac{\pi}{2}]$.

We now proceed to calculate the state fidelity given by $|\langle \hat \psi | \psi \rangle |^2$, where  $\ket{\hat \psi }$ is the ideal state given by Eq. (\ref{IdealStPr}).
Using the fact that the global phase on each subspace can be chosen freely, we can always set $ \langle \hat \psi_l | \psi_l \rangle \ge 0$, and therefore,
\begin{equation}
\langle \hat \psi | \psi \rangle  = \sum_l p_l \langle \hat \psi_l | \psi_l \rangle \ge \sum_l p_l |\langle \hat \psi_l | \psi_l \rangle |^2.
\end{equation}
Now, using the expressions of  $\ket{ \psi_l}$ and $  \ket{ \hat \psi_l} $     given by Eqs. (\ref{nonidealstate}) and (\ref{ProjIdealstate}), respectively, 
$\sum_l p_l |\langle \hat \psi_l | \psi_l \rangle |^2$ can be written as 
\begin{equation}
   \sum_l p_l |\langle \hat \psi_l \ket{ \psi_l}|^2
    = \sum_l p_l \frac{1}{2}|c^{(l)}_{000}-c^{(l)}_{111}|^2. \notag
\end{equation}
The expression in the right hand side of the above equation can be written as 
\begin{align}
   \frac{1}{2}|c^{(l)}_{000}-c^{(l)}_{111}|^2
    &=|c^{(l)}_{000}|^2+|c^{(l)}_{111}|^2-\frac{1}{2}|c^{(l)}_{000}+c^{(l)}_{111}|^2. \notag
    \end{align}
Using $\sum_{abc} |c^{(l)}_{abc}|^2=1$ in the first term in the right hand side of the above equation, we arrive 
at 
\begin{align}\label{stateFidel}
    \sum_l p_l |\langle \hat \psi_l \ket{ \psi}|^2 &= 1-  \sum_l p_l \sum_{abc \neq 000,111 } |  c^{(l)}_{abc}|^2 \notag \\
                                                   &-\sum_l p_l \frac{1}{2}|c^{(l)}_{000}+c^{(l)}_{111}|^2.
\end{align}
We will use the following lemma to obtain a lower bound on the right hand side of the above equation.
\begin{lemma}\label{prop: Anticommutativity}
	Suppose that the inequalities (\ref{eq: Error statistics n=3}) are satisfied for some $\epsilon>0$. Then, $\|\{ A_i,  B_i\}\ket{\psi}\|\le 4 \sqrt{2 \epsilon}$
	for all $i$.
\end{lemma}

\begin{proof}
%\textbf{Remik: This is a copy of a Lemma from Knill's paper. We should add citation.}
We show that $\|\{A_{1},B_{1}\}\ket{\psi}\|\le4\sqrt{2\epsilon}$.
From \eqref{eq: Error statistics n=3} and assuming that $0 \le \epsilon \le 1$, we have that
\begin{align}
    \| A_{1}\ket{\psi} + A_{2}A_3 {\ket{\psi}}\| &= \sqrt{2(1+\bra{\psi}A_1A_2A_3\ket{\psi})} \notag \\
    &\le \sqrt{2\left[1-(1-\epsilon)\right]}\notag\\
    &=\sqrt{2\epsilon}, \label{Difference: Column1}
\end{align}
and, similarly, we have 
\begin{align}
 \| B_{2}\ket{\psi} - B_{1}A_3 {\ket{\psi}}\| &\le \sqrt{2\epsilon}, \label{Difference: Column2} \\
 \| B_{1}\ket{\psi} - A_2B_3 {\ket{\psi}}\| &\le \sqrt{2\epsilon}, \label{Difference: Row1} \\
 \| B_{2}\ket{\psi} - A_1B_3 {\ket{\psi}}\| &\le \sqrt{2\epsilon}. \label{Difference: Row2}
\end{align}
Using then the triangle inequality for the vector norm 
and the fact that it is unitarily invariant, we 
have 
%Therefore, by a chain of triangle inequalities, and using the fact that
%$\|U\ket{\psi}\|=\|\ket{\psi}\|$ for any unitary $U$,
%
\begin{eqnarray}
\|(A_{1}B_{1}+B_{1}A_{1})\ket{\psi}\| &\le& \|(B_{1}A_{1}+B_{1}A_2A_{3})\ket{\psi}\| \nonumber\\
&&+ \|(-B_{1}A_2A_{3}+A_{2}B_2)\ket{\psi}\|\nonumber\\
&&+ \|(-A_{2}B_2+A_1A_2B_3)\ket{\psi}\| \nonumber\\
&&+ \|(-A_1A_2B_3+A_{1}B_{1})\ket{\psi}\| \nonumber\\
&\le& 4\sqrt{2\epsilon}.
\end{eqnarray}
Due to the symmetry of the inequality, the same will hold for any other $i$, which completes the proof.
\end{proof}
First, let us bound $\sum_l p_l |c^{(l)}_{000}+c^{(l)}_{111}|^2$ in Eq. (\ref{stateFidel}). 
From Eqs. \eqref{Difference: Column1} and \eqref{Difference: Column2},
respectively, we obtain
\begin{align} 
    &\sum_l p_l (|c^{(l)}_{000}+c^{(l)}_{111}|^2 + |c^{(l)}_{001}+c^{(l)}_{110}|^2 \nonumber \\
    &+|c^{(l)}_{010}+c^{(l)}_{101}|^2 + |c^{(l)}_{011}+c^{(l)}_{100}|^2) \le \epsilon, \nonumber
\end{align}
\begin{align} 
    &\sum_l p_l (|e^{-i\phi_l}c^{(l)}_{000}-e^{i\theta_l}c^{(l)}_{111}|^2 + |e^{i\phi_l}c^{(l)}_{010}+e^{i\theta_l}c^{(l)}_{101}|^2 \nonumber \\
    &+|e^{-i\phi_l}c^{(l)}_{001}-e^{i\theta_l}c^{(l)}_{110}|^2 + |e^{i\phi_l}c^{(l)}_{011}+e^{i\theta_l}c^{(l)}_{100}|^2) \le \epsilon. \nonumber
 \end{align}
From the first of these equations, it follows that
    \begin{align}
    \sum_l p_l |c^{(l)}_{000}+c^{(l)}_{111}|^2 &\le\epsilon.\label{eq: state fid proof 2}
\end{align}
It also follows that
\begin{align}
    \sum_l p_l |c^{(l)}_{001}+ c^{(l)}_{110}|^2 &\le \epsilon\label{eq: app coeff bound 1},\\
    \sum_l p_l |e^{-i\phi_l}c^{(l)}_{001}-e^{i\theta_l}c^{(l)}_{110}|^2&\le \epsilon.\label{eq: app coeff bound 2}
\end{align}
Next, we proceed to bound $\sum_l p_l \sum_{abc \neq 000,111 } |  c^{(l)}_{abc}|^2$ in Eq. (\ref{stateFidel}).
We have
\begin{align}
    &|c^{(l)}_{110}|^2|e^{i\theta_l}+e^{-i\phi_l}|^2\notag\\
    &= |(e^{i\theta_l}c^{(l)}_{110}-e^{-i\phi_l}c^{(l)}_{001})+e^{-i\phi_l}(c^{(l)}_{001}+c^{(l)}_{110})|^2.\notag
    \end{align}
    Using the fact that
$|x+y|^2\le2(|x|^2+|y|^2)$ for any $x,y\in\mathbb{C}$ in the above equation, we arrive at
    \begin{align}
     &|c^l_{110}|^2|e^{i\theta_l}+e^{-i\phi_l}|^2\notag\\
    &\le 2\big(|e^{i\theta_l}c^l_{110}-e^{-i\phi_l}c^{(l)}_{001}|^2 + |c^{(l)}_{001}+c^{(l)}_{110}|^2\big). \notag
\end{align}
From Eqs.~\eqref{eq: app coeff bound 1} and  \eqref{eq: app coeff bound 2},
\begin{align}\label{eq: ap eq1}
    &\sum_l p_l |c^{(l)}_{110}|^2|e^{i\theta_l}+e^{-i\phi_l}|^2\notag\\
    &= \sum_l p_l |c^{(l)}_{110}|^2(2+2 \cos(\theta_l+\phi_l))\le 4\epsilon.
\end{align}
Similarly,
\begin{align}
    &|c^{(l)}_{001}|^2|e^{i\theta_l}+e^{-i\phi_l}|^2\notag\\
    &= |(-e^{i\theta_l}c^{(l)}_{110}+e^{-i\phi_l}c^{(l)}_{001})+e^{i\theta_l}(c^{(l)}_{001}+c^{(l)}_{110})|^2\notag\\
    &\le 2\big(|-e^{i\theta_l}c^{(l)}_{110}+e^{-i\phi_l}c^{(l)}_{001}|^2 + |c^{(l)}_{001}+c^{(l)}_{110}|^2\big), \notag
\end{align}
from which it follows that
\begin{align}\label{eq: ap eq2}
    &\sum_l p_l |c^{(l)}_{001}|^2|e^{i\theta_l}+e^{-i\phi_l}|^2\notag\\
    &= \sum_l p_l |c^{(l)}_{001}|^2(2+2 \cos(\theta_l+\phi_l))\le 4\epsilon.
\end{align}
Adding Eqs.~\eqref{eq: ap eq1} and \eqref{eq: ap eq2},
\begin{equation}\label{eq: cos(th + ph)}
    \sum_l p_l(|c^{(l)}_{001}|^2+|c^{(l)}_{110}|^2)(1+\cos(\theta_l+\phi_l))\le 4\epsilon.
\end{equation}
We now need the following lemma, which is similar to Lemma~\ref{prop: Anticommutativity}.
\begin{lemma}\label{lemma: deg 4 norm bound}
Suppose the ideal expectations are satisfied to within error $\epsilon$. Then 	    $\|A_{1}B_{2}\ket{\psi}+B_{1}A_{2}\ket{\psi}\|\le 2 \sqrt{2 \epsilon}$.
\end{lemma}
\begin{proof}
\begin{align}
&\|A_{1}B_{2}\ket{\psi}+B_1A_2\ket{\psi}\|\le \|(A_{1}B_{2}-A_1B_1A_{3})\ket{\psi}\| \nonumber \\
&+ \|(A_1B_1A_{3}+B_{1}A_2)\ket{\psi}\|  \le 2\sqrt{2\epsilon}.\notag
\end{align}
where the last inequality follows from Eqs.~\eqref{Difference: Column1} and \eqref{Difference: Column2}.
\end{proof}

From the result of the lemma, we obtain
\begin{multline} \notag
    \sum_l p_l\big( (|c^{(l)}_{001}|^2+|c^{(l)}_{110}|^2)|e^{i\theta_l}+e^{i\phi_l}|^2\\
    + (|c^{(l)}_{011}|^2+|c^{(l)}_{100}|^2)|e^{i\theta_l}-e^{-i\phi_l}|^2 \big) \\
    +   \sum_l p_l\big( (|c^{(l)}_{000}|^2+|c^{(l)}_{111}|^2)|e^{-i\theta_l}+e^{-i\phi_l}|^2\\
    + (|c^{(l)}_{010}|^2+|c^{(l)}_{101}|^2)|e^{-i\theta_l}-e^{i\phi_l}|^2 \big) 
    \le 8 \epsilon,
\end{multline}
and therefore,
\begin{equation} \notag
    \sum_l p_l (|c^{(l)}_{001}|^2+|c^{(l)}_{110}|^2)|e^{i\theta_l}+e^{i\phi_l}|^2 \le 8\epsilon,
\end{equation}
or,
\begin{equation}\label{eq: cos(th - ph)}
    \sum_l p_l(|c^{(l)}_{001}|^2+|c^{(l)}_{110}|^2)(1+\cos(\theta_l-\phi_l))\le 4\epsilon.
\end{equation}
Adding Eqs.~\eqref{eq: cos(th + ph)} and \eqref{eq: cos(th - ph)},
\begin{align}
    \frac{8}{2}\epsilon &\ge \sum_l p_l(|c^{(l)}_{001}|^2+|c^{(l)}_{110}|^2)(1+\cos{\theta_l}\cos{\phi_l})\notag\\
    &\ge \sum_l p_l(|c^{(l)}_{001}|^2+|c^{(l)}_{110}|^2), \label{sumsingle1}
\end{align}
where in the last inequality we used $\theta_l$,$\phi_l\in[-\frac{\pi}{2}, \frac{\pi}{2}]$,
and so $\cos{\theta_l}\ge 0$ and $\cos{\phi_l}\ge 0$.
Similarly, from Eqs.  \eqref{Difference: Row1} and \eqref{Difference: Row2}, we have found the following bounds on $\sum_l p_l(|c^{(l)}_{010}|^2+|c^{(l)}_{101}|^2)$ and $\sum_l p_l(|c^{(l)}_{011}|^2+|c^{(l)}_{100}|^2)$:
\begin{align}
    \frac{8}{2}\epsilon &\ge \sum_l p_l(|c^{(l)}_{010}|^2+|c^{(l)}_{101}|^2), \label{sumsingle2}\\
    \frac{8}{2}\epsilon &\ge \sum_l p_l(|c^{(l)}_{011}|^2+|c^{(l)}_{100}|^2), \label{sumsingle3}
\end{align}
respectively.
Summing Eqs. (\ref{sumsingle1})-(\ref{sumsingle3}), we obtain 
\begin{align} \label{eq: state fid proof 1}
   \sum_l    p_l  \sum_{abc \neq 000,111 } |c^{(l)}_{abc}|^2&\le\frac{24}{2}\epsilon.
    \end{align}
Substituting Eqs. (\ref{eq: state fid proof 1}) and (\ref{eq: state fid proof 2}) in Eq. (\ref{stateFidel}), we obtain the state fidelity as 
%\begin{align}\label{Fidelity: State}
%     F(|\hat{\Psi}\rangle,\ket{\Psi})&\ge\left(1-\frac{25}{2}\epsilon\right)^2%\notag\\
%     &\ge 1 - 25 \epsilon.
%\end{align}
\begin{align}\label{Fidelity: State}
     F(|\hat \psi \rangle,\ket{\psi})&\ge\left(1-\frac{25}{2}\epsilon\right)^2\notag\\
     &\ge 1 - 25 \epsilon.
\end{align}

%Let us now generalize the above obtained bound on the state fidelity to any $ n \ge 3 $ in which case $\ket{\Psi_l}$ is expanded in the $Z$ basis  $\{\ket{n_1 n_2 \cdots n_n}: n_i \in \{0,1\}, \forall i \}$  within each Jordan space $\cH_l$ as

%Next we bound the fidelity of the operators in the case of $n=3$.
%From Eq.~\eqref{eq: jordan angles} and the definition of the ideal
%operators, for all $i$, $\tr(\hat{A}_i A_i)=8$ and so $F(\hat{A}_i, A_i)=1$.
%Eq.~\eqref{eq: jordan angles} implies
%\begin{equation} \notag
%    F(\hat{B}_{1}, B_{1}) = \sum_l p_l\cos{\theta_l}.
%\end{equation}
Next we bound the fidelity of the operators.
From Eq.~\eqref{eq: jordan angles} and the definition of the ideal
operators (up to an unitary freedom), it follows that for all $i$, $\tr(\hat{A}_i A_i)=8$ which implies that  $F(\hat{A}_i, A_i)=1$.
From Eq.~\eqref{eq: jordan angles}, it also follows that
\begin{equation} \label{fidel_B}
    F(\hat{B}_{1}, B_{1}) = \sum_l p_l\cos{\theta_l}.
\end{equation}
Let us now obtain  a lower bound on  $\sum_l p_l\cos{\theta_l}$.
Using both the result of Lemma~\ref{prop: Anticommutativity} and Eq.~\eqref{eq: jordan angles}, we obtain
\begin{align}\label{eq: Col 2 fid proof}
    4\epsilon &\ge \frac{1}{8}\|\{A_{1}, B_{1}\}\ket{\psi}\|^2 \notag\\
    &= \sum_l p_l\sin^2{\theta_l}.
\end{align}    
Using $\cos^2{\theta_l} + \sin^2{\theta_l}=1$ and $\sum_l p_l =1$, we can write     $\sum_l p_l\sin^2{\theta_l}$ as  
\begin{align}
  \sum_l p_l\sin^2{\theta_l}  &= 1 - \sum_l p_l\cos^2{\theta_l}\notag\\
    &\ge 1 - \sum_l p_l\cos{\theta_l},
\end{align}
where the inequality follows from
$\theta_l\in[-\frac{\pi}{2}, \frac{\pi}{2}]$ and therefore $\cos{\theta_l}\ge 0$.
From the above two equations, it follows  that 
\begin{equation}
 1-    \sum_l p_l\cos{\theta_l} \le 4 \epsilon.
\end{equation}
Using the above equation in Eq. (\ref{fidel_B}), we obtain
%\begin{equation*}
%    F(\hat{B}_{1}, B_{1})\ge1-4\epsilon.
%\end{equation*}
\begin{equation*}
    F(\hat{B}_{1}, B_{1})\ge1-4\epsilon.
\end{equation*}
%A similar calculation shows that $F(\hat{A}_{23}, A_{23})\ge1-\frac{25}{2}\epsilon$ also.
This ends the proof of Theorem \ref{thm robust}.

%\end{proof}

Theorem \ref{thm robust} can be straightforwardly extended to any $n$ as follows.  Here the state fidelity is defined as earlier and the operator fidelity
is defined with the different normalization factor as $F(\hat X_i,X_i):=(1/\dim(\hat{V}_n))\mathrm{Tr}(\hat X_i X_i)$, where $\dim(\hat{V}_n)$ is the dimension of the invariant subspace.
%and is similarly defined between $\hat B_j$ and $B_j$.}
%
\begin{thm}\label{thm robustn}
If a quantum state $|\psi \ra$ and a set of measurements $A_i,B_j$ with $i,j=\{ 1,2,\ldots,n \}$ in a Hilbert space $\cH_n$ satisfy the ideal expectations corresponding to the maximal quantum violation of the inequality \eqref{Ineqn} to within error     $\epsilon$, then there exists a projection $P:\cH_n \rightarrow \hat V_n$, where $\dim(\hat{V}_n)=2^n$, a state $\ket{\hat \psi}\in \hat{V}_n$ and $\hat A_i  $, $\hat{B}_j $ which are Hermitian involutions acting on $\hat{V}_n$ for all $i$ and $j$  such that 
\beq 
\la \hat \psi |  \hat A_1 \hat B_2 \hat B_3 \hat B_4... \hat B_n |\hat\psi\ra &=&  1 \nonumber \\ 
\la \hat \psi | \hat B_1 \hat A_2 \hat B_3 \hat B_4... \hat B_n |\hat\psi\ra &=&  1 \nonumber \\ 
... \nonumber \\
\la \hat  \psi | \hat B_1 \hat B_2 \hat B_3 \hat B_4... \hat A_n |\hat\psi\ra &=&  1 \nonumber
\eeq
and there also exists  an unitary $U$ acting on $ \hat V $ such that
\begin{align*}
     F(U |\hat \psi\rangle,\ket{\psi})&\ge 1- \epsilon_0,\\
    F( U \hat A_{i} U^{\dagger}, A_{i})&\ge 1- \epsilon_1 \quad \forall i, \\
    F( U \hat B_{i} U^{\dagger}, B_{i})&\ge 1- \epsilon_2 \quad \forall i , 
\end{align*}
where $\epsilon_0=[8(2^{n-1}-1)+1] \epsilon$, $\epsilon_1=0$, $\epsilon_2=2^{5-n} \epsilon$.
\end{thm}
%The above robustness also holds in the case of the temporal inequality provided that all ideal expectation values associated with this inequality are satisfied  to within error  $\epsilon$.
\begin{proof}
The bounds of this theorem are obtained using the similar steps used in the proof of Theorem \ref{thm robust}.
With respect to the Jordan decomposition,
\begin{equation} \label{Jdecn}
\cH_n=\bigoplus_l\cH_l, 
\end{equation}
where each $\cH_l$ has dimension at most $2^n$ and is 
invariant under the action of $A_{i}$ and $B_{j}$.
As before, the state $\ket{\psi}$ can be decomposed as 
\begin{equation}
    \ket{\psi}= \sum_l \sqrt{p_l}  \ket{\psi_l},
\end{equation} 
where  $\ket{\psi_l}  \in \cH_l$ 
and $\sum_l p_l =1$. With respect to the computational 
basis $\{\ket{n_1n_2 \cdots n_n}_l: n_i \in \{0,1\} \}$, we express each 
$\ket{\psi_l}$ as
\begin{equation} \notag
     \ket{ \psi_l}=\sum_{n_i \in\{0,1\}} c^{(l)}_{n_1n_2\cdots n_n}\ket{n_1 n_2 \cdots n_n}_l, 
\end{equation}
where 
%$\sum_{n_1n_2\cdots n_n}|c^l_{n_1n_2\cdots n_n}|^2=1$
$\sum_{n_1n_2\cdots n_n}|c_{n_1n_2\cdots n_n}|^2=1$. 

We define the subspace $\hat V_n \subseteq \cH_n $ as the linear span of $\{\ket{\tilde n_1 \tilde n_2 \cdots \tilde n_n} : = \sum_l \sqrt{p_l} \ket{n_1n_2 \cdots n_n}_l \}$.
We define the ideal state in this subspace  as
\begin{equation} \label{nIdealstate}
 \ket{\hat \psi }: =  \frac{1}{\sqrt{2}}\big(\ket{\tilde 0_1 \tilde 0_2 \cdots  \tilde 0_n  }-\ket{\tilde 1_1 \tilde 1_2 \cdots \tilde 1_n}\big),
\end{equation}
which can be re-expressed as 
\begin{equation}
    \ket{\hat \psi } =  \sum_l \sqrt{p_l}  \ket{\hat \psi_l},
\end{equation}
where 
\begin{equation}
\ket{\hat \psi_l}:=\frac{1}{\sqrt{2}}(\ket{0_10_2\cdots 0_n}_l-\ket{1_11_2\cdots 1_n}_l).
\end{equation}
Note that for the ideal observables defined as 
\begin{align}
\hat A_i&:=\bigoplus_l\left(\bigotimes_{i=1}^{k-1}\mathbbm{1}_l\otimes \hat A_i^l\otimes\bigotimes_{i=k+1}^{n}\mathbbm{1}_l\right), \nonumber \\
\hat B_j&:=\bigoplus_l\left(\bigotimes_{i=1}^{k-1}\mathbbm{1}_l\otimes \hat B_j^l\otimes\bigotimes_{i=k+1}^{n}\mathbbm{1}_l\right), \nonumber 
\end{align}
where
%\begin{align}\label{eq: jordan angles}
%    \hat{A}_{1}&= XII,\ \hat{B}_{2}= (\cos \phi_l IYI+\sin \phi_l IXI),\notag\\
%    \hat{A}_{2}&= IXI,\ \hat{B}_{1}= (\cos \theta_l YII+\sin \theta_l XII),  \notag\\
%    \hat{A}_{3}&= IIX,\ \hat{B}_{3}= (\cos \nu_l IIY+\sin \nu_l IIX),
%\end{align}
%with $\theta_l,\phi_l, \nu_l \in[-\frac{\pi}{2},\frac{\pi}{2}]$ and $X$ and $Y$ are the Pauli operators.
\begin{align}
    \hat A_{i_l}&= X_i,\quad \hat  B_{j_l}= Y_j,\notag
    \end{align}
    with $i,j=1,2,3$, and, for $i,j=4,5,\cdots n$,
\begin{align}
   \hat  A_{i_l}&= -Y_i,\quad \hat  B_{j_l}= X_j, \notag
\end{align}
the ideal state defined in Eq. (\ref{nIdealstate}) violates the noncontextuality inequality (\ref{Ineq3}) maximally.

% Note that the following set of quantum observables acting on $\mathbbm{C}^{2^n}$: 
%$A_j=X_j$ and $B_j=Y_j$ with $j=1,2,3$, and,  $A_k=-Y_k$ and $B_k=X_k$, with $k=4,5,\cdots n$ 
%and the GHZ state expressed as
%\begin{equation}
%    \ket{\psi}=\frac{1}{\sqrt{2}}(\ket{0000 \cdots 0} -  \ket{1111 \cdots 1}  )
%\end{equation}
%can be used to obtain the maximal violation of our noncontextuality inequality (\ref{Ineqn}). We extend this realization to the invariant subspace $V_n$ to compare 
%the ideal realization with the nonideal realization. We express the ideal realization acting on $V_n$ using 
%the $Z$ basis  on each $\cH_l$: $\{\ket{n_1 n_2 \cdots n_n}_l: n_i \in \{0,1\}, \forall i \}$ as
%\begin{equation} \notag
% \ket{ \hat \psi }  = \sum_l  \sqrt{\frac{p_l}{2}}\big(\ket{00 \cdots 0}_l-  \ket{11 \cdots 1}_l\big).
%\end{equation} 
%On the other hand, $\ket{ \psi}$ leading to the given nonmaximal quantum violation of the inequality (\ref{Ineqn}) 
%is expressed as
%\begin{equation} \notag
%     \ket{\Psi_l}=\sum_{n_i \in\{0,1\}} c_{n_1n_2\cdots n_n}^l\ket{n_1 n_2 %\cdots n_n}, 
%\end{equation}
%\begin{equation} \notag
%     \ket{ \psi}=\sum_l \sqrt{p_l} \sum_{n_i \in\{0,1\}} c^{(l)}_{n_1n_2\cdots n_n}\ket{n_1 n_2 \cdots n_n}_l, 
%\end{equation}
%where 
%$\sum_{n_1n_2\cdots n_n}|c^l_{n_1n_2\cdots n_n}|^2=1$
%$\sum_{n_1n_2\cdots n_n}|c_{n_1n_2\cdots n_n}|^2=1$. 

From Appendix \ref{Appendix: Jordan}, it follows that the nonideal observables can be written as
\begin{align}
A_i&=\bigoplus_l\left(\bigotimes_{i=1}^{k-1}\mathbbm{1}_l\otimes A_i^l\otimes\bigotimes_{i=k+1}^{n}\mathbbm{1}_l\right), \nonumber \\
B_j&=\bigoplus_l\left(\bigotimes_{j=1}^{k-1}\mathbbm{1}_l\otimes B_j^l\otimes\bigotimes_{j=k+1}^{n}\mathbbm{1}_l\right), \nonumber 
\end{align}
for all $i,j=1,2,\cdots n$.
We choose unitary such that the operators ${A}_{i_l}$ and $B_{j_l}$ acting on  $\cH_l$  can be written as follows:
\begin{align}\label{eq: jordan angles_n}
    A_{i_l}&= X_j, \quad   B_{j_l}=  \cos\theta_{j_l} Y_j+\sin \theta_{j_l} X_j, \quad i,j=1,2,3, \\
    B_{i_l}&= X_k, \quad   A_{j_l}=  \cos\theta_{j_l} Y_k+\sin \theta_{j_l} X_k, \quad i,j=4,5, \cdots n,
\end{align}
with $\theta_{j_l} \in[-\frac{\pi}{2},\frac{\pi}{2}]$ for all $j$ and $k$.

Using the similar steps used to obtain Eq. (\ref{stateFidel}), we have
\begin{align}\label{FidStateprior}
    \langle \hat \psi \ket{ \psi} &\ge 1- \sum_l p_l \sum_{n_1n_2\cdots n_n \neq 0_10_2\cdots 0_n,1_11_2\cdots 1_n } |c^{(l)}_{n_1n_2\cdots n_n}|^2 \nonumber \\
    &-\frac{1}{2}\sum_l p_l|c^{(l)}_{0_10_2\cdots0_n}+ c^{(l)}_{1_11_2\cdots1_n}|^2.
\end{align}
Similarly to the case of $n=3$, a bound on the right hand side of the above equation can be obtained as follows.
The term $\frac{1}{2}\sum_l p_l|c^{(l)}_{0_10_2\cdots0_n}+ c^{(l)}_{1_11_2\cdots1_n}|^2$ in  Eq. (\ref{FidStateprior}) can be bounded using the inequality given  by
\begin{align}
    \| A_{1}\ket{\psi} + A_{2}A_3B_4\cdots B_n{\ket{\psi}}\| \le \sqrt{2\epsilon}. \label{Difference: Column1n}
\end{align}
From this equation, we obtain
\begin{align} 
    &\sum_l p_l (|c^{(l)}_{0_10_20_3\cdots 0_n}+ c^{(l)}_{1_11_21_3\cdots 1_n}|^2 + \cdots   \\ \nonumber
    &+ |c^{(l)}_{0_11_21_3\cdots 1_n}+ c^{(l)}_{1_10_20_3\cdots 0_n}|^2) \le \epsilon, \nonumber
\end{align}
from which it follows that 
\begin{equation}\label{boundnfid1}
\sum_l p_l|c^{(l)}_{0_10_2\cdots0_n}+ c^{(l)}_{1_11_2\cdots1_n}|^2 \le \epsilon.
\end{equation}

Next, the second term     in Eq. (\ref{FidStateprior})  can be bounded using the other inequalities such as
\begin{align}
 \| B_{2}\ket{\psi} - B_{1}A_3B_4\cdots B_n {\ket{\psi}}\| &\le \sqrt{2\epsilon}, \label{Difference: Column2n}
\end{align}
from which we obtain,
\begin{align} 
    &\sum_l p_l \big(\sum_{n_4,\cdots, n_n}|e^{-i\theta_{2_l}}c^{(l)}_{000n_4 \cdots  n_n}- e^{i\theta_{1_l}}c^{(l)}_{111 \bar n_4 \cdots \bar n_n}|^2    \nonumber \\ 
    &+\sum_{n_4,\cdots, n_n}|e^{i \theta_{2_l}}c^{(l)}_{010n_4 \cdots n_n}+ e^{i\theta_{1_l}}c^{(l)}_{101 \bar n_4\cdots \bar n_n}|^2 \nonumber \\
    &+\sum_{n_4,\cdots, n_n}|e^{-i \theta_{2_l}}c^{(l)}_{001n_4 \cdots n_n}- e^{i \theta_{1_l}}c^{(l)}_{110\bar n_4 \cdots \bar n_n}|^2  \nonumber \\
    &+\sum_{n_4,\cdots, n_n}|e^{i \theta_{2_l}}c^{(l)}_{011n_4\cdots n_n}+e^{i \theta_{1_l}}c^{(l)}_{100\bar n_4\cdots \bar n_n}|^2\big)   \le \epsilon, \nonumber
 \end{align}
 where $\bar n_i$, with $i=4,5,\cdots n$, denotes $n_i \oplus_2 1$.
%\begin{align} 
%    &\sum_l p_l (|e^{-i\theta_{2_l}}c^{(l)}_{0000 \cdots 0}- e^{i\theta_{1_l}}c^{(l)}_{1111 \cdots 1}|^2 + \cdots \nonumber \\ 
%    &+|e^{-i \theta_{2_l}}c^{(l)}_{0001 \cdots 1}-e^{i \theta_{1_l}}c^{(l)}_{1110 \cdots 0}|^2  \nonumber \\ 
%   &+|e^{i \theta_{2_l}}c^{(l)}_{0100 \cdots 0}+ e^{i\theta_{1_l}}c^{(l)}_{1011\cdots 1}|^2+\cdots \nonumber \\
%     &+|e^{i \theta_{2_l}}c^{(l)}_{0101 \cdots 1}+e^{i \theta_{1_l}}c^{(l)}_{1010\cdots 0}|^2 \nonumber \\
%    &+|e^{-i \theta_{2_l}}c^{(l)}_{0010\cdots 0}- e^{i \theta_{1_l}}c^{(l)}_{1101 \cdots 1}|^2 + \cdots \nonumber \\
%     &+|e^{-i \theta_{2_l}}c^{(l)}_{0011\cdots 1}- e^{i \theta_{1_l}}c^{(l)}_{1100 \cdots 0}|^2  \nonumber \\
%    &+|e^{i \theta_{2_l}}c^{(l)}_{0110\cdots 0}+e^{i \theta_{1_l}}c^{(l)}_{1001\cdots 1}|^2) +\cdots   \nonumber \\
%    &+|e^{i \theta_{2_l}}c^{(l)}_{0111\cdots 1}+e^{i \theta_{1_l}}c^{(l)}_{1000\cdots 0}|^2)  \le \epsilon. \nonumber
% \end{align}
From the above equation, using the steps similar to the ones used to obtain the bound given by Eq. (\ref{sumsingle1}), we obtain a bound on $\sum_l p_l(|c^{(l)}_{0010 \cdots 0}|^2+|c^{(l)}_{1101 \cdots 1}|^2)$ as follows: 
\begin{align}\label{pairsumbound}
 \sum_l p_l(|c^{(l)}_{0010 \cdots 0}|^2+|c^{(l)}_{1101 \cdots 1}|^2) \le    \frac{8}{2}\epsilon.
\end{align}
The sum $\sum_{n_1n_2\cdots n_n \neq 00\cdots 0,11\cdots 1 } \sum_l p_l|c^{(l)}_{n_1n_2\cdots n_n}|^2$ can be spilt into the sum  of $(2^{n-1}-1) $ terms which are a sum of modulus of two coefficients $c^{(l)}_{n_1n_2\cdots n_n}$  as in the left hand side of Eq. (\ref{pairsumbound}). These  $(2^{n-1}-1) $ terms have the same bound as given in Eq. (\ref{pairsumbound}). Therefore, we obtain
\begin{align}\label{boundnfid2}
   \sum_l p_l   \sum_{n_1n_2\cdots n_n \neq 00\cdots 0,11\cdots 1 }  |c^{(l)}_{n_1n_2\cdots n_n}|^2&\le\frac{8(2^{n-1}-1)}{2}\epsilon.
    \end{align}
 Using Eqs. (\ref{boundnfid1}) and (\ref{boundnfid2}) in Eq. (\ref{FidStateprior}), we obtain  the bound on the fidelity as given in Theorem. \ref{thm robustn}.

Next, we bound the fidelity of the operators. As in the case of $n=3$, we also have $\|\{A_{1},B_{1}\}\ket{\psi}\|\le4\sqrt{2\epsilon}$ which implies that
\begin{align}
    2^{5-n} \epsilon &\ge \frac{1}{2^n}\|\{A_{1}, B_{1}\}\ket{\psi}\|^2 \notag\\
    &\ge 1 - \sum_l p_l\cos{\theta_l},
\end{align}
leading to the following bound on the fidelity between $\hat{B}_{1}$ and %$B_{1}$
${B}_{1}$:
%\begin{equation*}
%    F(\hat{B}_{1}, B_{1})\ge1-2^{5-n}  \epsilon,
%\end{equation*}
\begin{equation*}
    F(\hat{B}_{1}, B_{1})\ge1-2^{5-n}  \epsilon,
\end{equation*}
employing the similar steps as in the case of $n=3$. This ends the proof of Theorem \ref{thm robustn}.
\end{proof}

\section{Jordan's lemma} \label{Appendix: Jordan}

In this section we prove a corollary to Jordan's lemma which is
a direct generalization of Corollary 7.1 proven in Ref. \cite{IMO+20}.
For completeness we also state Jordan's lemma (see, e.g., Ref. \cite{IMO+20} for a proof).

\begin{lemma}[Jordan's lemma]\label{lemma: Jordan}
    Let $A$ and $B$ be a pair of Hermitian operators acting on a Hilbert space $\cH$ 
    such that $A^2=B^2=\mathbbm{1}$. Then, $\cH$ decomposes as a direct sum $\cH=\bigoplus_l\cH_l$, with $\dim \cH_l\in\{1,2\}$,
    and $A$ and $B$ act invariantly on each $\cH_l$.
\end{lemma}

In this way, since the set of eigenvectors of $AB$ span $\mathcal{H}$, we can decompose the Hilbert space $\mathcal{H} = \bigoplus_l \mathcal{H}_l$ where the dimension of such $\mathcal{H}_l$ is at most $2$.

\begin{coro}\label{corrNonideal}
Let $A_i$ and $B_{j}$ with $i=1,2,3$ be Hermitian operators acting 
on a Hilbert space $\mathcal{H}$ that square to indetity and
satisfy the following commutation relations
\begin{equation}\label{commutApp}
    [A_i,A_j]=[A_i,B_j]=0 \qquad (i\neq j).
\end{equation}
Then, $\cH$ can be decomposed as 
\begin{equation}
\cH=\bigoplus_l (\cH_{l}^1\otimes\cH_{l}^2 \otimes\cH_{l}^3),
\end{equation}
where each local Hilbert space $\mathcal{H}_l^i$ is of dimension at most two.
Moreover,
$A_1=\bigoplus_l (A_{1_l}\otimes\mathbbm{1}_l \otimes\mathbbm{1}_l )$, 
$B_1=\bigoplus_l (B_{1_l}\otimes\mathbbm{1}_l \otimes\mathbbm{1}_l )$,
$A_{2}=\bigoplus_l (\mathbbm{1}_l\otimes A_{2_l} \otimes\mathbbm{1}_l )$,
$B_{2}=\bigoplus_l (\mathbbm{1}_l\otimes B_{2_l} \otimes\mathbbm{1}_l )$,
$A_{3}=\bigoplus_l (\mathbbm{1}_l\otimes\mathbbm{1}_l\otimes A_{3_l}  )$,
and
$B_{3}=\bigoplus_l (\mathbbm{1}_l \otimes\mathbbm{1}_l \otimes B_{3_l}  )$.
\end{coro}

\begin{proof}This proof is a direct generalization of that of Corollary 7.1 proven in Ref. \cite{IMO+20}.

First, let us notice that Eq. (\ref{commutApp}) implies that $[A_iB_i, A_jB_j]=0$ for any $i,j=1,2,3$, which means that all three Hermitian operators $A_iB_i$ can be jointly diagonalised. 

Let then $\ket{\alpha,\beta,\gamma}$ be an eigenvector of 
these operators such that $A_1B_1\ket{\alpha,\beta,\gamma}=\alpha\ket{\alpha,\beta, \gamma}$
and $A_2B_2\ket{\alpha,\beta,\gamma}=\beta\ket{\alpha,\beta,\gamma}$ and $A_3B_3\ket{\alpha,\beta,\gamma}=\gamma\ket{\alpha,\beta,\gamma}$.
Define the following vectors 
%(by convenience, in the braket notation) 
$\ket{\bar\alpha,\beta,\gamma}=A_1\ket{\alpha,\beta,\gamma}$,
$\ket{\alpha,\bar\beta, \gamma}=A_2\ket{\alpha,\beta,\gamma}$, $\ket{\alpha,\beta, \bar\gamma}=A_3\ket{\alpha,\beta,\gamma}$,
$\ket{\bar\alpha,\bar\beta,\gamma}=A_1A_2\ket{\alpha,\beta,\gamma}$, $\ket{\bar\alpha,\beta,\bar\gamma}=A_1A_3\ket{\alpha,\beta,\gamma}$,
$\ket{\alpha,\bar\beta,\bar\gamma}=A_2A_3\ket{\alpha,\beta,\gamma}$
and
$\ket{\bar\alpha,\bar\beta,\bar\gamma}=A_1A_2A_3\ket{\alpha,\beta,\gamma}$.
Then, the subspace 
\begin{eqnarray}
 \mathrm{span}\{\ket{\alpha,\beta,\gamma}, \ket{\bar\alpha,\beta,\gamma}, \ket{\alpha,\bar\beta,\gamma},  \ket{\alpha,\beta,\bar\gamma},\nonumber\\ \ket{\bar\alpha,\bar\beta,\gamma},
\ket{\bar\alpha,\beta,\bar\gamma},\ket{\alpha,\bar\beta,\bar\gamma},\ket{\bar\alpha,\bar\beta,\bar\gamma}\}   
\end{eqnarray}
is isomorphic to $\mathrm{span}\{\ket{\alpha}, \ket{\bar\alpha}\}\otimes\mathrm{span}\{\ket{\beta}, 
\ket{\bar\beta}\} \otimes\mathrm{span}\{\ket{\gamma}, \ket{\bar\gamma}\} $.
It follows from  Lemma~\ref{lemma: Jordan} that
both $A_1, B_{1}$, both $A_2$, $B_{2}$ as well as both $A_3$, $B_{3}$  
act invariantly on the first, second and third tensor factors,
and trivially on the others, respectively.
\end{proof}

%\textbf{Below is a copy of the paper by Knill.}
%We remark that in the main text we assume that each subspace $\cH_l$ has dimension 8. This is done without loss of
%generality, since we can extend any smaller dimensional
%subspace to $8$ dimensions, with all operators acting trivially on the extension. 

%\textcolor{blue}{Maybe we remove the above text and add a statement as in above Eq. 105 in the main text.}

The above corollary can be trivially generalized to any $n \ge 3$: 
let $A_i$ and $B_{j}$ $(i,j=1,\ldots,n)$ be Hermitian operators acting on $\mathcal{H}_n$ that square to the identity and satisfy 
\begin{equation}
    [A_i,A_j]=[B_i,B_j]=0\qquad  (i\neq j).
\end{equation}
Then, $\cH_n$ decomposes
as 
$\cH_n=\bigoplus_l (\cH_l^1\otimes\cH_l^2  \cdots \otimes \cH_l^n )$,
with $\dim\cH_{l}^i\leq 2$, and
\begin{equation}
    A_j=\bigoplus_l\left(\bigotimes_{i=1}^{j-1}\mathbbm{1}_l\otimes A_j^l\otimes\bigotimes_{i=j+1}^{n}\mathbbm{1}_l\right)
\end{equation}
and
\begin{equation}
    B_j=\bigoplus_l\left(\bigotimes_{i=1}^{j-1}\mathbbm{1}_l\otimes B_j^l\otimes\bigotimes_{i=j+1}^{n}\mathbbm{1}_l\right).
\end{equation}

%$A_1=\bigoplus_l (A_{1}^l\otimes_{n-1}\mathbbm{1}_l )$, 
%$B_1=\bigoplus_l (B_{1}^l\otimes_{n-1}\mathbbm{1}_l )$,
%$A_{2}=\bigoplus_l (\mathbbm{1}_l\otimes A_{2}^l \otimes_{n-2} \mathbbm{1}_l )$,
%$B_{2}=\bigoplus_l (\mathbbm{1}_l\otimes B_{2}^l \otimes_{n-2}\mathbbm{1}_l )$,
%and the other operators $A_i$ and $B_j$ have a similar decomposition.

\end{document}